\newif\iffundamversion
\newif\ifarxivversion

\fundamversionfalse 

\iffundamversion
  \arxivversionfalse
\else
  \arxivversiontrue
\fi

\iffundamversion
    \documentclass{fundam}
\else
    \documentclass{article}
    \usepackage[margin=4.3cm]{geometry}
    \usepackage{amsmath}
\fi

\usepackage{amssymb}
\usepackage{eufrak}

\usepackage{hyperref}
\hypersetup{
    colorlinks = false
}

\usepackage{xspace}
\usepackage{complexity}
\usepackage{tikz}
\usetikzlibrary{positioning,matrix}
\usetikzlibrary{decorations.markings}
\usetikzlibrary{arrows}

\renewcommand{\AM}            {\ensuremath{\mathcal{AM}}}
\newcommand{\AMo}             {\ensuremath{\AM^{0}}}
\newcommand{\AMowoDiss}       {\ensuremath{\AMo\mkern-9mu{}_{-d}}}
\newclass  {\MC}              {MC}
\newcommand{\MCsemi}          {\ensuremath{\MC^*}}
\newcommand{\PMCuni}          {\ensuremath{\P\MC}}
\newcommand{\PMCsemi}         {\ensuremath{\P\MCsemi}}
\newcommand{\no}{\texttt{\textit{no}}}
\newcommand{\yes}{\texttt{\textit{yes}}}

\newcommand{\EvolveRule}[3]{\ensuremath{\left[\,#1\,\to\,#3\,\right]_{#2}}} %
\newcommand{\CommInRule}[3]{\ensuremath{#1\left[\,\right]_{#2}\,\to\,\left[\,#3\,\right]_{#2}}} %
\newcommand{\CommOutRule}[3]{\ensuremath{\left[\,#1\,\right]_{#2}\,\to\,\left[\,\right]_{#2}\,#3}} %
\newcommand{\DivideRule}[4]{\ensuremath{\left[\,#1\,\right]_{#2}\,\to\,\left[\,#3\,\right]_{#2}\left[\,#4\,\right]_{#2}}} %

\newcommand{\arule}{\ensuremath{\mathrm{a}}}            
\newcommand{\brule}{\ensuremath{\mathrm{b}}}            
\newcommand{\crule}{\ensuremath{\mathrm{c}}}            
\newcommand{\drule}{\ensuremath{\mathrm{d}}}            
\newcommand{\erule}{\ensuremath{\mathrm{e}}}

\newfunc{\MS}{MS}
\newfunc{\parent}{parent}
\newcommand{\env}{env}
\newcommand{\membraneConfig}{\mathcal{C}}
\newcommand{\encodedMemSystem}{\ensuremath{\langle \Pi \rangle}}

\newcommand{\setOfObjects}{\ensuremath{O}}

\newcommand{\setOfMultisets}{\ensuremath{M}}
\newcommand{\setOfLabels}{\ensuremath{H}}
\newcommand{\setOfIdLabelRels}{\ensuremath{\Lambda}}
\newcommand{\setOfMembraneRules}{\ensuremath{R}}
\newcommand{\membraneFamily}{\mathbf{\Pi}}

\newcommand{\ORgate}{\textsc{Or}}
\newcommand{\ANDgate}{\textsc{And}}
\newcommand{\NOTgate}{\textsc{Not}}

\newclass{\DLOGTIME}{DLOGTIME}
\newclass{\FL}{FL}
\newclass{\FAC}{FAC}
\newclass{\tally}{tally}
\newclass{\length}{length}
\newlang{\Parity}{Parity}

\newcommand{\Reducible}[3]{\ensuremath{#2_{\mathrm{#1}}(#3)}}
\newcommand{\Reduces}[2]{\leq^{#2}_{\mathrm{#1}}}

\newcommand{\bigO}[1] {\ensuremath{\mathcal{O}(#1)}}
\newcommand{\nats}{\ensuremath{\mathbb{N}}}
\newcommand{\tallyMachine}{\ensuremath{\mathcal{T}}}

\newcommand{\semiencoder}{\ensuremath{h}}
\newcommand{\encoder}{\ensuremath{e}}
\newcommand{\former}{\ensuremath{f}}

\newcommand{\muniform}      [3]{\ensuremath{(#1,#2)\textrm{-uniform-}#3}}
\newcommand{\uniform}      [2]{\ensuremath{#1\textrm{-uniform-}#2}}
\newcommand{\semiunif}      [2]{\ensuremath{#1\textrm{-semi-uniform-}#2}}

\newcommand{\ACACuniPMCowodiss}{\ensuremath{\muniform{\FAC^0}{\FAC^0}{\PMCuni\AMowoDiss}}}
\newcommand{\ACPMCwodissSemi}{\ensuremath{\semiunif{\FAC^0}{\PMCsemi\AMowoDiss}}}
\newcommand{\UniformMembrane}{\ACACuniPMCowodiss}
\newcommand{\SemiUniMembrane}{\ACPMCwodissSemi}

\newcommand{\emphatic}{recogniser\ensuremath{_{\,\geqslant 1}}\xspace}

\newcommand{\pgfextractangle}[3]{%
    \pgfmathanglebetweenpoints{\pgfpointanchor{#2}{center}}
                              {\pgfpointanchor{#3}{center}}
    \global\let#1\pgfmathresult  
}

\ifarxivversion
  \usepackage{authblk}
  \newcommand{\email}[1]{\url{#1}}
  \usepackage{amsthm}
  \newtheorem{theorem}{Theorem}
  \newtheorem{definition}[theorem]{Definition}  
  
  \newtheorem{lemma}[theorem]{Lemma}

\fi

\usepackage{authblk}

\begin{document}

\title{Uniformity is weaker than semi-uniformity for some membrane systems}
\renewcommand\Authand{ \hspace{.1\textwidth} }

\newcommand{\nmthanks}{\thanks{%
   \href{mailto:a-nimurp@microsoft.com}{a-nimurp@microsoft.com}
   Supported by a PICATA Young Doctors fellowship from
    CEI Campus Moncloa, UCM-UPM, Madrid, Spain and an Embark Fellowship from the Irish Research Council for Science, Engineering and Technology.}}
 \newcommand{\dwthanks}{\thanks{%
 \href{mailto:woods@caltech.edu}{woods@caltech.edu}.
 Supported by NASA grant NNX13AJ56G, and NSF grants 0832824 \& 1317694 (The Molecular Programming Project), CCF-1219274 and CCF-1162589}}

\author[1]{Niall~Murphy\nmthanks}
\author[2]{Damien Woods\dwthanks}
\affil[1]{Microsoft Research Cambridge}
\affil[2]{California Institute of Technology}
 
\date{}
\maketitle

\begin{abstract}
 We investigate computing models that are presented as families of finite computing devices with a uniformity condition on the entire family.
  Examples of such models include Boolean circuits, membrane systems, DNA computers, chemical reaction networks and tile assembly systems, and there are many others. 
  However, in such models there are actually two distinct kinds of uniformity condition. 
  The first is the most common and well-understood, where each input length is mapped to a single computing device (e.g.\ a Boolean circuit) that computes on the finite set of inputs of that length.
  The second, called semi-uniformity, is where each input is mapped to a computing device for that input (e.g.\ a circuit with the input encoded as constants).
  The former notion is well-known and used in Boolean circuit complexity, while the latter notion is frequently found in literature on nature-inspired computation from the past 20 years or so.
  
  Are these two notions distinct? For many models it has been found that these notions are in fact the same, in the sense that the choice of uniformity or semi-uniformity leads to characterisations of the same complexity classes. 
  In other related work, we showed that these notions are actually distinct for certain classes of Boolean circuits. 
  Here, we give analogous results for membrane systems by showing that certain classes of uniform membrane systems are  strictly weaker than the analogous semi-uniform classes. This solves a known open problem in the theory of membrane systems. We then go on to present results towards characterising the power of these semi-uniform and uniform membrane models in terms of $\NL$ and languages reducible to the unary languages in $\NL$, respectively.
\end{abstract}
\iffundamversion
\fi

  \section{Introduction}
  Many of the early DNA computing algorithms~\cite{HRBBLS2000,Lip1995,LWFCCS2000,OKLL1997}  involved mapping
  an instance of an $\NP$-hard problem (such as Maximal Clique) to a
  set of DNA strands and lab protocols, and then using well-known
  biomolecular techniques to solve the problem.  To assert generality for
  such an algorithm one would define a mapping from arbitrary problem
  instances to sets of DNA polymers and experimental protocols.  In order to
  claim that this mapping is not doing the essential computation, it would
  have to be  easily computable (for example,~logspace computable).  Circuit
  uniformity (introduced by Borodin~\cite{Bor1977}) provides a
  well-established framework where we map each input length~$n\in\mathbb{N}$
  to a circuit~$C_n \in \mathfrak{C}$, with a suitably simple mapping.  However, some of the
  DNA computing algorithms cited above do something different, they map an
  \emph{instance} $x$ of the problem to a computing device $C_x$ that is
  unique  to that input (via a suitably simple encoding function).
  This latter notion is called \emph{semi-uniformity}~\cite{Pau2003,PRS2003p},
  and in fact quite a number
  of nature-inspired computational models use semi-uniformity. 
  This raises
  the immediate question of whether the notions of uniformity and
  semi-uniformity are computationally equivalent. 
  We investigate this question in the field of  membrane computing or
  P-systems~\cite{Pau2003,PRS2010_handbook}.
  This is a branch of natural computing which explores the power of
  computational models that are inspired by the structure and function of
  living cells. 

  It has been shown in a number of models that whether one chooses to use
  uniformity or semi-uniformity does not affect the power of the model.
  However, our main result shows that uniformity is computationally strictly 
  weaker than semi-uniformity for a number of classes of  membrane systems. 
  Specifically, we prove that choosing one notion over another in this setting
  gives characterisations of complexity classes that are known to be distinct.
   The uniform versus semi-uniform question that we address has
  been stated as Open Problem C in~\cite{Pau2005c}.  

  Why is this result surprising?
  We know that the class of problems solved by a uniform family of devices
  is contained in the analogous semi-uniform class, since the former is a
  restriction of the latter.
  However, it is also the case that in almost all membrane system models
  studied to date, the classes of problems solved by semi-uniform and uniform
  families turned out to be equal, see, e.g.,~\cite{AP2007c,MW2007p,Sos2003p}.
  Specifically, if we want to solve some problem, by specifying a family of
  membrane systems (or some other model), it is often much easier to first use
  the more general notion of semi-uniformity, and then subsequently try to
  find a uniform solution.
  In almost all cases where a polynomial time semi-uniform family of membrane systems was given for some
  problem~\cite{AP2004p,PRS2003p,Sos2003p}, at a later point a uniform
  version of the same result was published~\cite{AMP2003p,AP2007c,PRS2003p}.
  Here we prove that this improvement is not always possible.

  We go on to give a number of other results that tease out the computational
  power of semi-uniform and uniform families of membrane systems.
  
  Our main result proves something general about uniform and semi-uniform families of finite devices that is
  independent of particular models and formalisms. Our techniques can be applied to other
  computational models besides membrane systems and 
  we have demonstrated this by showing similar results for Boolean circuits~\cite{MurphyWoods2013}.
  Indeed, a  number of other models  explicitly, or implicitly, use notions of
  uniformity and semi-uniformity. 
  Models presented as uniform families of devices include 
    membrane systems and 
    Boolean circuits 
      as noted above, as well as 
    DNA computers~\cite{Adl1994p,BSRW2009,WS2005p,qian2011scaling,Winfree05bar},
    chemical reaction networks~\cite{Cook2009,CHMT2012,TC2012,thachuk2013space},
    neural networks~\cite{Par1994b}
    and other models studied in computational complexity theory.
  Besides membrane systems, a surprising number of models, including some just mentioned, 
  are presented as semi-uniform
  families of devices, including 
    DNA computers~\cite{Lip1995,LWFCCS2000},
    chemical reaction networks~\cite{Cook2009,TC2012},
    the abstract tile assembly model~\cite{rothemund2000program,SW2007p},
    the nubots model of active  molecular self-assembly and robotics~\cite{nubots, nubotsDNA20},
    and an insertion-based polymer model~\cite{dabbyChenSODA2012,MalchikWinsow}. 
    Uniform and semi-uniform families of devices are both 
    natural ways to present a model of computation and 
    elucidating the distinction between them seems a worthy goal.
    
  Furthermore, although we do not formally show it, our results hold for a
  version of the stochastic chemical reaction network
  model~\cite{soloveichik2008computation} that meets our definitions for 
  membrane systems and in particular where there are families of networks
  deciding languages and \emph{unimolecular reactions only} (in the model there 
  are discrete natural number molecular counts and all reactions are of the form $A \rightarrow
  \mathcal{M} $, where $\mathcal{M}$ is a mutliset of molecular species).
  Interestingly, these results also hold if we generalise this model to use
  maximally parallel synchronous reaction updates. This shows that adding the
  seemingly strong and unrealistic ability of maximal parallelism in this
  context conveys no extra power to the model
    (despite the fact that it does increase the power of more general,
    bimolecular for example,
    chemical reaction network models).

  Our main result is of importance to work on models of computation and
  natural computing since it highlights that the (seemingly harmless) choice
  between uniformity and semi-uniformity in these models may lead to drastic
  changes in computational power.  How drastic? Roughly speaking, we find that
  the semi-uniform models studied here characterise the class $\NL$, while the
  analogous uniform models have power comparable to, or more formally
  reducible to, the unary languages in $\NL$.  Our work here and on Boolean
  circuits suggests that this question should be asked of other computational
  models.

    \subsection{Overview of results}
    Roughly speaking, a membrane system consists of a membrane-bound
    compartment  that contains other (possibly nested) membrane-bound
    compartments that in turn contain \emph{objects} that interact with each
    other and with membranes to carry out a computation. A family, or set, of
    \emph{recogniser} membrane systems decides a language $L$. Families can be
    uniform or semi-uniform. For a uniform family there is an associated pair
    of functions $(f,e)$, where $f$ maps a binary input word $x$, of length
    $n$, to a membrane system~$\Pi_n$ that may be used to process any word of
    length $n$, and $e$ encodes $x$ as a multiset of input objects to $\Pi_n$
    (for each of the $2^n$ words of length $n \in\mathbb{N}$ we have a single
    membrane system $\Pi_n$). For a semi-uniform family, a single function $h$
    maps the input word $x$ to a membrane system $\Pi_x$ (for each word we
    have a membrane system). In either case, rules are applied to objects in
    the membrane system until it produces special object(s) indicating that
    $x$ is accepted or rejected. Of course the encoding functions $f,e,h$
    should be suitably simple so that the membrane system, and not the
    encoding functions, are doing the interesting work. In this paper we use
    $\FAC^0$ uniformity and semi-uniformity, that is, the functions $f,e,h$
    are in $\FAC^0$, the class of functions computed by uniform constant depth
    polynomial size Boolean circuits; this is a class of fairly simple
    problems and is mostly known for what it does \emph{not} contain.

    In Section~\ref{sec:uni_noteq_semi} we give our main result, 
      that  
        uniform families of active membrane systems 
          without charges and dissolution (denoted $\AMowoDiss$)
        that run in polynomial time are strictly weaker than their semi-uniform counterpart.  
    We prove this by showing that these uniform families solve no more than non-uniform-$\AC^0$, 
        a class that does not even contain $\Parity$
          (the set of words over $\{0,1\}$ with an odd number of 1s).
    The analogous semi-uniform systems can indeed solve $\Parity$ and do much else besides.
    In fact, for two out of three models that we consider, the semi-uniform
    families exactly characterise $\NL$. This is shown in
    Section~\ref{sec:semi_uniform} and illustrated as
    Theorem~\ref{thm:PMCwodiss-is-NL} at the top of
    Figure~\ref{fig:summaryofallresults}. 
  
    This leaves the question:  what is the exact power of uniform families of $\AMowoDiss$ systems?
    In previous papers, where 
    more powerful membrane systems and complexity classes are studied,
    e.g.~\cite{AMP2003p,AP2004p,AP2007c}, model definitional choices 
    were not so important. 
    In our setting, definitional details such as the choice of uniformity 
    condition and the particular kinds of acceptance modes allowed 
      for such \emph{recogniser} membrane systems 
        lead to seemingly different results and some open questions as we now describe. 
    
    We give results for three variants on the definition of recogniser
    membrane system. The most powerful are \emph{acknowledger} membrane
    systems, where an accepting computation should produce one or more~$\yes$
    objects, and a rejecting computation should produce zero $\yes$ objects.
    In Section~\ref{sec:membrane-uni} we give an exact characterisation of
    uniform families of acknowledger $\AMowoDiss$ membrane systems. It turns
    out that they decide exactly those languages that are $\FAC^0$ disjunctive
    truth-table reducible to the unary languages in~$\NL$ (called
    $\tally\NL$). See Theorem~\ref{thm:uni_is_dtt_tallyNL} in
    Figure~\ref{fig:summaryofallresults}.

    In Section~\ref{sec:stricter} 
    we consider \emph{\emphatic} membrane systems: a restriction of
    acknowledger systems where an accepting computation produces one or more
    $\yes$ objects and zero $\no$ objects, and a rejecting computation
    produces one or more $\no$ objects and zero $\yes$ objects.
    We give upper and lower-bounds, in terms of classes reducible to
    $\tally\NL$, for uniform families of \emph{\emphatic} $\AMowoDiss$
    systems. In Figure~\ref{fig:summaryofallresults}, two upper bounds are
    illustrated as Theorems~\ref{lem:unimem_subset_dttTallyNL}
    and~\ref{lem:strict_unimem_in_ctt_tallynl}, and a lower bound as
    Theorem~\ref{lem:strict_mTallyNL_subset_uniMem}.

    The more standard, uniform \emph{recogniser} systems,  are a restriction of
    \emphatic membrane systems and are defined so that an accepting computation
    should produce a single $\yes$ object and zero $\no$ objects, and a rejecting
    computation should produce a single $\no$ object and zero $\yes$ objects. As
    noted above, our results (Figure~\ref{fig:summaryofallresults},
    Theorem~\ref{lem:uni_less_semi}) show that these uniform recogniser systems
    are strictly weaker than semi-uniform recogniser systems in our setting. We do
    not give a tight characterisation for the power of uniform recogniser systems,
    but discuss this as an open problem in Section~\ref{sec:openquestions}.
        
  \begin{figure}[t]
    \centering
  \resizebox{\linewidth}{!}{%
    \begin{tikzpicture}[%
      equiv/.style={%
        double distance=0.45ex,double,thick,
        -,
        >=stealth',
      },
      strict/.style={
        very thick,
        ->,
        >=stealth',
        decoration={%
          markings,
          mark=at position 0.5 with {%
            \node [#1] {$\subsetneq$};
          }
        },
        postaction={decorate}
      },
      refer/.style 2 args={
        decoration={%
          markings,
          mark=at position 0.5 with {%
            \node [inner sep=2pt,#2] {(\ref{#1})};
          }
        },
        postaction={decorate}
      },
      include/.style={->,>=stealth',very thick}]

     \node (tallyNL) {$\tally\NL$};
     \node (mTallyNL) [above=2em of tallyNL] {$\Reducible{m}{\FAC^0}{\tally\NL}$};
     \node (ttTallyNL) [above=8em of mTallyNL] {$\Reducible{tt}{\FAC^0}{\tally\NL}$};
     \node (TTallyNL) [above=4ex of ttTallyNL] {$\Reducible{T}{\FAC^0}{\tally\NL}$};
     \node (nl) [above=2em of TTallyNL] {$\NL$};
     \node (semiMem)  [left=1.5em of nl] {$\ACPMCwodissSemi$};
    \node (ackrecnote) [above=-1ex of semiMem] {acknowledger};

     \node (semiMemEmphatic)  [right=1.5em of nl] {$\ACPMCwodissSemi$};
    \node (emphaticnode) [above=-1ex of semiMemEmphatic] {\emphatic};

     \node (dttTallyNL)    [above=5em of mTallyNL, xshift=-4.7em]
        {$\Reducible{dtt}{\FAC^0}{\tally\NL}$};
     \node (cttTallyNL)    [above=5em of mTallyNL, xshift=4.7em]
        {$\Reducible{ctt}{\FAC^0}{\tally\NL}$};

     \node (uniformMem) [left=1.5em of dttTallyNL] {$\UniformMembrane$};
     \node (acknote) [above=-1ex of uniformMem] {acknowledger};
     \draw [equiv={xshift=-1.7ex},refer={thm:uni_is_dtt_tallyNL}{yshift=2.2ex}] (uniformMem) -- (dttTallyNL);

     \node  (mUniformMOR)  [above=2.6em of mTallyNL] {$\star$};
     \node (memlabel) [below=3.4em of cttTallyNL,xshift=8em] {$\UniformMembrane$};
     \node (recnote) [above=-1ex of memlabel] {\emphatic};
     \draw [dashed] (memlabel) -- (mUniformMOR); 
     
     \draw [include,
            refer={lem:strict_mTallyNL_subset_uniMem}{xshift=3ex}]
              (mTallyNL) -- (mUniformMOR); 
     \draw [include,
            refer={lem:unimem_subset_dttTallyNL}{xshift=-1.0ex,yshift=-1.5ex}] 
              (mUniformMOR) -- (dttTallyNL); 
     \draw [include,
            refer={lem:strict_unimem_in_ctt_tallynl}{xshift=1.0ex,yshift=-1.5ex}] 
              (mUniformMOR) -- (cttTallyNL); 
 
      \draw [equiv={yshift=-1.4ex},
             refer={thm:PMCwodiss-is-NL}{yshift=2.2ex}] 
              (semiMem) -- (nl); 
      \draw [equiv={yshift=-1.4ex},
             refer={thm:PMCwodiss-is-NL}{yshift=2.2ex}] 
              (nl) -- (semiMemEmphatic); 
      \draw [strict={xshift=-1.9ex,rotate=90}] (TTallyNL) -- (nl); 

      \draw [include] (dttTallyNL) -- (ttTallyNL); %
      \draw [include] (cttTallyNL) -- (ttTallyNL);
      \draw [strict={xshift=-1.9ex,rotate=90}] (tallyNL) -- (mTallyNL);
      \draw [include] (ttTallyNL) -- (TTallyNL);
      
      \draw [include] (mTallyNL) edge[out=30,in=320,->] (cttTallyNL);
      \draw [include] (mTallyNL) edge[out=150,in=210,->] (dttTallyNL);

      \pgfextractangle{\angle}{acknote}{semiMem}
      \draw [strict={xshift=-2.5ex,rotate=\angle},
             refer={lem:uni_less_semi}{xshift=3.5ex,yshift=0ex}]
             (acknote) -- (semiMem);
      
    \end{tikzpicture} 
  }
  \caption{%
    Summary of results.
        Numerical labels refer to theorems which are proved in this paper,
      and symbols are used to show inclusion type,
        with an unlabelled arrow denoting $\subseteq$.
    The figure shows the relationship between~$\NL$, 
      $\tally\NL$ (the set of unary languages decided in non-deterministic
        logspace), 
        and a number of classes $\Reducible{r}{\FAC^0}{\tally\NL}$ of languages 
        that are reducible to $\tally\NL$ by various types of $\FAC^0$ computable reductions $\text{r}$.
    The star ($\star$) indicates the class 
    labeled by the dashed line.
    See~\cite{MurphyWoods2013} for proofs of inclusions that do not have numerical labels and for more on classes reducible to $\tally\NL$.  }
  \label{fig:summaryofallresults}
  \end{figure}
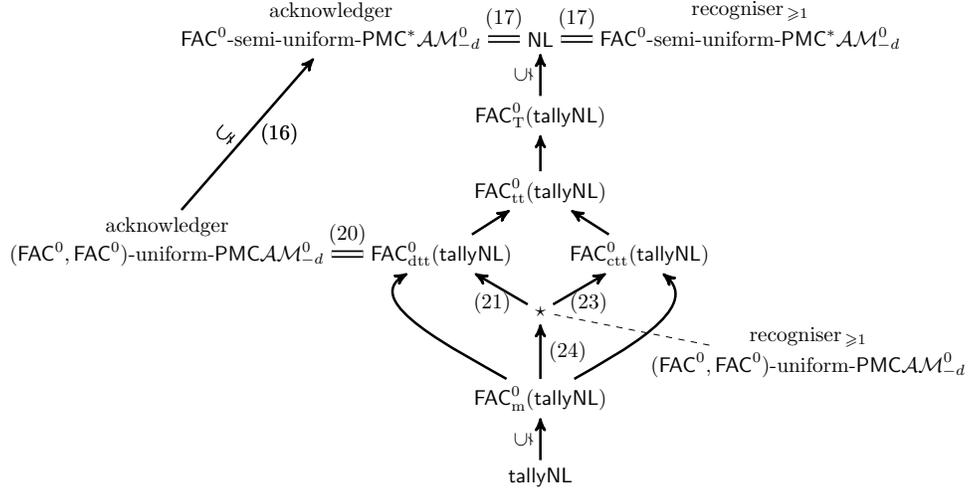
  
    We note that there is a previous $\P$ characterisation 
      for both uniform and semi-uniform families of
        active membrane systems without charges and
        dissolution~\cite{NJNC2006p}: the same systems as we  use here, but
        under much more general uniformity conditions, namely \emph{polynomial time}, or $\P$, uniformity. 
    In that work the authors are motivated by the relationship with classes
    above $\P$ and so it is sufficient in their work to use $\P$ uniformity.
    When using significantly
    tighter uniformity conditions (e.g.\ $\FAC^0$), such polynomial-time encoding
    functions for uniform and semi-uniform families can be seen to be stronger
    than the membrane systems themselves~\cite{MW2011}  (assuming $\NL \subsetneq \P$).
    In this paper we use $\FAC^0$~uniformity which is weak enough to expose the
    underlying  power of certain, suitably weak,  classes of active membrane systems without charges or
    dissolution.
    A number of other varieties of membrane systems (e.g.~\cite{GPR2009,PP2010}) 
      also claim $\P$ characterisations that depend on $\P$ uniformity. We leave it as a possible direction for 
      future work to investigate these, and other,
      membrane systems under suitably tight notions of uniformity or semi-uniformity. 

  \section{Definitions} 
    \label{sec:complexity_defs}
    For a function $f \colon \{0,1\}^* \mapsto \{0,1\}^*$ and integers $m,n \geq 1$
    let $f_n \colon \{0,1\}^n \mapsto \{ 0,1\}^{m}$ be the restriction of $f$ to
    domain and range consisting  of strings of length $n$ and $m$ respectively. 
    We consider only functions~$f$ where for each $n$ there is an $m$ such that all
    length-$n$ strings in $f$'s domain are mapped to length-$m$ strings, 
    thus~$f = \bigcup_{n = 0}^{\infty} f_n$.  
    Each language $L \subseteq \{ 0,1\}^{\ast}$ has an
    associated total \emph{characteristic function}
    $\chi_L \colon \{ 0,1\}^{\ast} \mapsto \{0, 1\}$ defined  by $\chi_L (w) = 1$
    if~$w \in L$ and~$0$ if~$w \notin L$.
    We say a language $L$ is decided by a Turing machine $M$ if $M$ computes the 
    characteristic function $\chi_L$.
    For a string~$w$, we let~$|w|$ denote its length.
    
    Let $\NL$ be the class of languages accepted  
    by non-deterministic logarithmic-space Turing machines. Such machines have a 
    read-only input tape, a write-only output tape and a read-write work tape 
    whose length is a logarithmic function of  input length. 
    The class of functions 
    computed by a deterministic logarithmic-space Turing machines (with an
    additional write-only output tape) is denoted $\FL$.
  
    Let $\tally$ be the set of all languages over the one-letter alphabet $\{ 1\}$.
    We define $\tally\NL = \tally \cap \NL$, i.e.\ the class of all tally
    languages and length encoded languages in $\NL$. 
    For more details on complexity classes and Turing machines see~\cite{Pap1993x}. 
  
    A circuit $C_n$ computes a function 
      computes a function $f_n\colon\{0,1\}^n \mapsto \{0,1\}^m$ 
      on a fixed number~$n$ of Boolean variables.
    We consider functions of an arbitrary  number of variables by defining
    (possibly infinite) families of circuits.
    We say a family of
      circuits~$\mathfrak{C} = \{ C_n \mid n \in \mathbb{N} \}$ 
      computes a function~$f \colon \{0, 1\}^* \mapsto \{0, 1\}^*$
      if for all $n \in \mathbb{N}$
      and for all $w\in\{ 0,1\}^n$
      circuit $C_n$ outputs the string $f_n(w)$. 
    We say a family of circuits $\mathfrak{C}$ decides 
      a language~$L \subseteq \{ 0,1\}^*$ if for each~$w\in\{0,1\}^n$
      circuit~$C_n \in \mathfrak{C}$ on input $w$ computes~$\chi_L$. 
  
    In a \emph{non-uniform} family of circuits there is no required  
    similarity between family members.    
    In order to specify such a requirement   we use a \emph{uniformity
      function} that algorithmically specifies the similarity between members
    of a circuit family.
    Roughly speaking, a \emph{uniform circuit family}~$\mathfrak{C}$ is an
    infinite sequence of circuits with an associated function~$\former\colon
    \{1\}^* \rightarrow \mathfrak{C}$ that generates members of the family and
    is computable within some resource bound.  
    For more details on Boolean circuits see~\cite{Vol1999x}. 

    When dealing with uniformity for small complexity classes one of the 
    preferred uniformity conditions is $\DLOGTIME$-uniformity~\cite{BIS1990p}.
    Roughly speaking, 
      a circuit is $\DLOGTIME$-uniform if 
        there is a procedure that can decide 
          if a word is in the ``connection language'' of the circuit family 
          in time linear in the word length. 
      Each word of the connection language encodes 
        either 
          an input gate of the circuit, 
          an output gate of the circuit,
          or a wire connecting the output of one identified gate to the input
          of a second identified gate. 
      Each word also encodes, in binary, the number $n$ for this circuit.
      For more details on $\DLOGTIME$ uniformity see~\cite{AK2010,BIS1990p}.
    
    The \emph{depth} of a circuit is the length of the longest path from an input
      gate to an output gate.
    The \emph{size} of a circuit is the number of wires it contains~\cite{AK2010}.
      
    Non-uniform-$\AC^0$ is the set of languages decidable by 
    families of constant-depth polynomial-size
      (in input length~$n$) circuits with unbounded
      fan-in $\ANDgate$ and $\ORgate$ gates, and $\NOTgate$ gates with fan-in`1.
    $\AC^0$ is the set of languages decidable by constant-depth polynomial-size
      (in input length~$n$) \iffundamversion\linebreak\fi $\DLOGTIME$-uniform circuits with unbounded
      fan-in $\ANDgate$ and $\ORgate$ gates, and $\NOTgate$ gates with fan-in~1.
    $\FAC^0$ is the class of functions computable by polynomial-size
      constant-depth $\DLOGTIME$-uniform circuits with unbounded fan-in
      $\ANDgate$ and $\ORgate$ gates, and $\NOTgate$ gates with fan-in~1.

    \subsection{Reductions} 
    For concreteness, we explicitly define some standard types of reductions. 
    Let $A, B \subseteq \{ 0,1\}^{\ast}$.
    Let $\C$ be a set of functions (for example $\FL$ or $\FAC^0$), 
      a function $f$ is $\C$-computable if $f \in \C$.

    \begin{definition}[Many-one reducible]
      \label{def:many_One_reduction}
      Set $A$ is many-one reducible to set $B$, written $A \Reduces{m}{\C} B$,
      if there is a function $f$ that  is $\C$-computable with the
      property that for all $w$, $w \in A$, if and only if $f(w) \in B$.
    \end{definition}

    The following definition of truth table reduction comes
    from~\cite{BB1988,BHL1995}, see also~\cite{LLS1975,post1944}.    
    The Boolean function $\sigma$ is historically called a truth table~\cite{post1944}.
  
    \begin{definition}[Truth-table reduction]
      \label{def:truthtablereduction}
      Set $A$ is $\C$ truth table reducible to set $B$,
      written $A \Reduces{tt}{\C} B$,
      if there exists $\C$-computable functions~$\tau : \{ 0,1\}^* \rightarrow \{ 0,1\}^* \times \{ 0,1\}^* \times \ldots \,\,\, $  
        and~$\sigma : \{ 0,1\}^* \rightarrow \{ 0,1\}$ 
        such that 
        $w \in A$ if and only if 
        $\tau(w) = (a_1, \ldots, a_{\ell_w})$ such that 
        $\sigma( \chi_B(a_1),\ldots, \chi_B(a_{\ell_w}) ) = 1$,
        where
        $\chi_{B}$ is the characteristic function of $B$.
    \end{definition}
    A \emph{disjunctive} truth table reduction (dtt) is one where 
      at least one string generated by $\tau(w)$ is in~$B$,  in other words  $\sigma( \chi_B(a_1),\ldots, \chi_B(a_{\ell_w}) ) = \bigvee_{1\leq i \leq \ell_w} \chi_B(a_i)$.
    A \emph{conjunctive} truth table reduction (ctt) is one where 
      all the strings generated by $\tau(w)$ are in $B$, in other words  $\sigma( \chi_B(a_1),\ldots, \chi_B(a_{\ell_w})) = \bigwedge_{1\leq i \leq \ell_w} \chi_B(a_i)$. 
      
    \begin{definition}[Turing reducible]
      \label{def:turing_reduction}
      Set $A$ is $\C$ Turing reducible to $B$, written $A \Reduces{T}{\C} B$, if
      there is a Turing machine $M$, that is resource-bounded in the same way machines computing functions in $\C$ are, such that
      $w \in A$ iff~$M$ accepts~$w$ with~$B$ as its oracle. %
    \end{definition}

    The following implications follow directly from these definitions, 
      for more details see~\cite{LLS1975}.
      \begin{equation*}
        \begin{tikzpicture}[baseline]
          \matrix (m) [row sep=-1em, column sep=-0.8em]{%
            & \node[rotate=15]{$\implies$}; & \node{$A \Reduces{dtt}{\C} B$}; & \node[rotate=-15]{$\implies$}; & \\ %
             \node{$A \Reduces{m}{\C} B$}; & & & & \node{$A \Reduces{tt}{\C} B
               \implies A \Reduces{T}{\C} B$ };\\ %
            & \node[rotate=-15]{$\implies$}; & \node{$A \Reduces{ctt}{\C} B$}; & \node[rotate=15]{$\implies$}; & \\
          };
        \end{tikzpicture}
      \end{equation*}

    Let $\Reducible{r}{\FAC^0}{\C}$ be the set of all languages that are
      $\FAC^0$ reducible to languages in $\C$ via a reduction of some type  
      $\mathrm{r} \in \{ \text{m, dtt, ctt, tt, T} \}$. 

    \subsection{Configuration graphs}
    \label{sec:config_graphs} 
    
    \begin{definition}[Configuration Graph]
     Let $w \in \{0,1\}^*$ be the input to a halting $s(|w|)$-space bounded Turing
     machine~$M$. %
      The \emph{configuration graph} $G_{M,w}$ of $M$ on input $w$ is an acyclic directed graph 
        where for each potential configuration of $M$ there is a vertex that encodes it and
        where
        a potential configuration consists of  
          an input read bit,
          work tape contents,
          input tape head position
          and work tape head position.
      The graph $G_{M,w}$ has a directed
          edge from a vertex~$c$ to a vertex $c'$ if the configuration encoded by $c'$ 
          can be reached from the configuration encoded by~$c$ in one
          step via $M$'s transition function.
    \end{definition}

    A configuration graph $G_{M,w}$ has the property that there is a directed
      path from the vertex $c_s$ representing the initial configuration, to the accept vertex $c_a$ 
      if an only if $M$ accepts input~$w$. 
      Also, we consider only space bounded Turing machines that do not repeat
      a configuration (i.e.\ loop), hence we define configuration graphs to be acyclic which will be a useful property later on.
    We are interested in $\bigO{\log |w|}$ space bounded Turing machines, whose configuration graphs are of size (number of vertices) $\bigO{|w|^2 |Q|}$.
    Lemma~\ref{lem:config_in_AC0} follows from Theorem~3.16~in~\cite{Imm1999x}.
    \begin{lemma}
      \label{lem:config_in_AC0}
      Given the binary encoding of 
        a Turing machine $M$, which has state set $Q$ and
        an $\FAC^0$ computable space bound $\bigO{\log |w|}$,
        and given an input $w$, 
        the configuration graph $G_{M,w}$, of size $\bigO{|w|^2 |Q|}$, 
        is computable in $\uniform{\DLOGTIME}{\FAC^0}$.
    \end{lemma}
        
    \subsection{Membrane systems}
    \label{sec:mem_def}
    In this section we define the specific variant of membrane systems we use 
    in this paper.
    We also define recognizer membrane systems, uniform families and some complexity classes.
    These definitions are based on those from the literature~\cite{GLPZ2013,PRRW2009x}.
    
    In this paper the term \emph{membrane systems} and the notation 
      $\AMowoDiss$ refer to active membrane systems without charges and without
      dissolution rules~\cite{NJNC2006p,PRRW2009x}. 

    Let $\MS(O)$ represent the set of all multisets over the elements of the
    finite set $O$. 

    \begin{definition}
      \label{def:membrane}
      A \emph{membrane system} of type $\AMowoDiss$ %
      is a tuple~$\Pi =
      (\setOfObjects,  \mu,  \setOfMultisets, \setOfLabels,\setOfIdLabelRels,
      \setOfMembraneRules)$ where:
      \begin{itemize}
        \item $\setOfObjects$ is the alphabet of objects (or the set of object
          types); 
        \item $\mu = (V_\mu, E_\mu, \env)$ is a rooted tree representing the
          membrane structure. 
          $V_\mu \subsetneq \nats$ is the finite set of membranes.
          $E_\mu \subsetneq V_\mu \times V_\mu$
              such that $(p,c) \in E_\mu$ if
                the (parent) membrane $p$ contains the (child) membrane $c$.
            The root, $\env \in V_\mu$, of the tree is the  only 
              membrane with no parent and is called the ``environment''.
            Leaves of the tree represent ``elementary membranes'': i.e.\
            membranes  which contain no  other membranes.  
        \item $\setOfMultisets\colon V_\mu \rightarrow \MS(\setOfObjects)$ map each 
          membrane to an object multiset, defining the membrane's object contents;
        \item $\setOfIdLabelRels\colon  V_\mu \rightarrow \setOfLabels$ is an
          injective mapping of membranes to $\setOfLabels$, the finite set of membrane labels.
          In this work the environment membrane always has the label ``$\env$''; 
        \item $\setOfMembraneRules$ is a finite set of developmental rules of
          the following types (where~$o,u,c \in \setOfObjects$ and $w \in
          \MS(\setOfObjects)$, $h \in \setOfLabels$):
          \begin{itemize}
            \item [($\arule$)] $\EvolveRule{o}{h}{w}$ (object rewriting), an
              object $o$ in a membrane with label $h$ is replaced by a multiset
              of objects $w$. 
            \item [($\brule$)] $\CommInRule{o}{h}{u}$ (communication in),
              an object $o$ in a membrane with a child membrane with label $h$ is
              moved into the child membrane and modified to become $u$.
            \item [($\crule$)] $\CommOutRule{o}{h}{u}$  (communication out),
              an object $o$ in a membrane with label $h$ is
              moved into the parent membrane and modified to become $u$.
            \item [($\erule$)]  $\DivideRule{o}{h}{u}{v}$ (elementary membrane division),
              an elementary membrane with label $h$ containing object $o$ is duplicated, 
              in one copy $o$ is replaced by $u$ while  
                in the other copy it is replaced by $v$.
          \end{itemize}
      \end{itemize}
        The environment membrane cannot divide nor communicate out objects.%
        \footnote{Definitions of active membranes often include a second container
        membrane that cannot dissolve called the ``skin''~\cite{PRRW2009x},
        we omit this from our definitions. 
        The proofs in this paper can be easily modified to account for a skin.} 
    \end{definition}
    The missing ($\drule$) rule is the dissolution rule   
         which we do not consider in this paper. 
    Active membrane systems may also have \emph{non-elementary membrane division}
    rules~\cite{PRRW2009x}.
    That is,  membranes with child membranes may also divide. 
    For the kinds of membrane systems we consider in this paper the inclusion or omission of non-elementary division 
    rules does not affect the results~\cite{NJNC2006p,MW2011}.

    A \emph{configuration}~$\membraneConfig$ of a membrane system is a
    tuple~$(\mu = (V_\mu, E_\mu, \env) , \setOfMultisets, \setOfIdLabelRels)$ whose elements are defined
    in Definition~\ref{def:membrane} (with the exception that $\setOfIdLabelRels$
    may be non-injective).

    A \emph{permissible encoding} of a membrane system~$\encodedMemSystem$, or of
    a configuration~$\langle \membraneConfig \rangle$, encodes all multisets in a
    unary manner.  For example, a multiset is encoded in the format~$\left[
    a,a,a,b,b\right]$, rather than in the shorter format $a^3 b^2$.
    Likewise, the membrane structure should be encoded such that each
    membrane child-parent relation is written explicitly. 

    A configuration $\membraneConfig_i$ \emph{transitions}  to configuration
    $\membraneConfig_{i+1}$ by the application of a multiset of rules $\mathcal{R}$ from the set
    $\setOfMembraneRules$.
    The rules are applied in a \emph{maximally parallel manner}.  That
    is, at each timestep, a multiset of applicable rules $\mathcal{R}$ is
    non-deterministically chosen such that 
      (i) all rules in $\mathcal{R}$ are applicable, and
      (ii) there does not exist a multiset of applicable rules $\mathcal{R}'$ such that 
      $\mathcal{R} \subsetneq \mathcal{R'}$. 
    Rules are \emph{applicable} in a timestep according to the following principles:
        Rules are applied to the most deeply nested membranes first. 
        In each timestep, an object can be involved in at most one rule of
        any type.
        A membrane can be the subject of at most one rule of type 
        ($\brule$), ($\crule$) or ($\erule$).
        If a membrane is divided (a rule of type ($\erule$)) and
        there are objects in this membrane which evolve via rules of type
        (\arule), then we assume that first the  type (\arule) rules  are applied,
        and then the division rule. 
        All other rules are applied non-deterministically. 

    A \emph{computation} of a membrane system is a sequence of
    configurations where each configuration  transitions to the next. 
    As noted above, at a given timestep the multiset of applicable rules is
    non-deterministically chosen: therefore on a given input there are multiple
    possible computations. In other words, membrane systems are
    non-deterministic. 
    A computation that reaches a configuration
    where no more rules are applicable is called a \emph{halting computation}.

  \subsubsection{Recogniser, \emphatic, and acknowledger membrane systems}\label{sec:recog}
For the following three definitions it is the case that the set of objects $\setOfObjects$ contains the special objects $\yes$ and $\no$ and that there are no rules applicable to $\yes$ or $\no$ (hence if $\yes$ or $\no$ are created, they can never be destroyed).
  The standard~\cite{PRRW2009x} definition of a recogniser membrane system is as follows.
  \begin{definition}
    \label{def:recogniser}
  A \emph{recogniser membrane system} is a membrane system  such that 
           all computations halt,
           and at the halting step (and not before) 
           exactly one of the objects
           $\yes \in \setOfObjects$ 
           or~$\no \in \setOfObjects$
           appears in the multiset of the environment membrane.
  \end{definition}
   A computation that halts
    with $\yes$ in the environment is referred to as an \emph{accepting computation}
    while one with $\no$ in the environment is referred to as a
    \emph{rejecting computation}. In this paper, and in previous
    work~\cite{MW2008bc,MW2011}, we also use the following more general
    systems:

  \begin{definition}
    \label{def:emphaticrecogniser}
    A \emph{\emphatic  membrane system} is a membrane system  such that 
           all computations halt, and
            either 
              (a) one or more copies of the object 
                $\yes \in \setOfObjects$ 
            or 
              (b) one or more copies of the object 
                $\no \in \setOfObjects$
           appear in the multiset of the environment membrane, but not both. %
  \end{definition}
  As with recogniser membrane systems, a computation of a \emphatic
     membrane system that halts
     with~$\yes$ in the environment is referred to as an \emph{accepting computation}
    while one with $\no$ in the environment is referred to as a \emph{rejecting computation}. In this paper we also use the following systems that are more general than the two above:

  \begin{definition}
    \label{def:acknowledger}
 \emph{Acknowledger membrane systems} are systems such that 
    all computations halt 
    (and where one or more copies of the distinguished object $\yes$ 
      may or may not appear in the $\env$ membrane). 
  \end{definition}
    We say that a computation of an acknowledger membrane system is an \emph{accepting
    computation}
    if at least one $\yes$ object is present in the $\env$ membrane at the final step.
  A computation of an acknowledger membrane system is in a \emph{rejecting computation}
    if  there are zero $\yes$ objects in 
    the $\env$ membrane at the final step.

\subsubsection{Families of membrane systems}
    There are two main notions of uniformity considered in the membrane
    computing literature defined as follows.
    
    \begin{definition}[Semi-uniform families]
      \label{def:semiuniform}
      A family of membrane systems systems 
      $\membraneFamily = \{\Pi_w \mid w \in \Sigma^* \}$
      is said to be \emph{semi-uniform} if there is a function 
      $\semiencoder \colon \Sigma^* \mapsto \membraneFamily$ 
      that maps from each input word $w$ to a description (in a permissible encoding) of a membrane system~$\Pi_w$. 
    \end{definition}

    \begin{definition}[Uniform families]
      \label{def:uniform}
      A family of membrane systems 
        $\membraneFamily = \{\Pi_{n} \mid n \in \mathbb{N}  \}$ 
        is said to be uniform if there are two associated functions:
        \begin{enumerate} 
        \item $\former \colon 1^* \mapsto \membraneFamily$  that maps $1^{n}$ (the unary representation of $n$) %
            to the description (in a permissible encoding) of a membrane
            system $\Pi_{n}$ with a designated input membrane; 
      \item  $\encoder \colon \Sigma^* \mapsto \MS(\setOfObjects)$
      that maps a word $w\in \Sigma^*$ to the \emph{input} multiset $\encoder(w)$ (in a permissible
          encoding) where~$\setOfObjects$ is the set of objects of
          $\former(1^n)$, $n = |w|$.
          \end{enumerate}
    \end{definition}
       We let   $\Pi_{n}(\encoder(w))$ denote  the membrane system $\former(1^{n}) = \Pi_{n}$ with 
        the multiset $\encoder(w)$ in its designated input membrane. Note that both $\Pi_{n}$ and $\Pi_{n}(\encoder(w))$ are membrane systems.

    In this paper, we deal only with confluent membrane systems:
    in a \emph{confluent membrane system}~$\Pi$ all computations of $\Pi$
    agree on the answer, that is, either all of $\Pi$'s computations are
    accepting (in which case $\Pi$ accepts) or else all of $\Pi$'s computations are rejecting (in which case $\Pi$ rejects). 

    A semi-uniform family, $\membraneFamily$, 
    recognises a language $L \subseteq \Sigma^\ast$ confluently  
    if for all $w\in\Sigma^\ast$ there exists 
      $\Pi_w \in \membraneFamily$ such that
        $w \in L$ implies that  $\Pi_w$ accepts confluently 
        and $w \notin L$ implies $\Pi_w$ rejects confluently. 
    A uniform family, $\membraneFamily$, with encoder $\encoder$, 
    recognises a language $L \subseteq \Sigma^\ast$ confluently  
    if for all $w\in\Sigma^\ast$ there exists 
      $\Pi_{|w|}(\encoder(w))$ where $\Pi_{|w|} \in \membraneFamily$ such that 
        $w \in L$ implies that  $\Pi_{|w|}(\encoder(w))$ accepts confluently 
        and $w \notin L$ implies $\Pi_{|w|}(\encoder(w))$ rejects confluently. 
    Such a (semi-)uniform families are called a \emph{confluent} families of 
     recogniser, \emphatic, or acknowledger membrane systems. 
      
    That is, each
    membrane system $\Pi$ in a confluent family starts from a fixed initial
    configuration 
    and then $\Pi$ non-deterministically chooses one
    from a number of valid computations.  All of
    these valid computations give the same result: either all accepting (if
    $w\in L$) or else all rejecting (if $w\notin L$). 

    If the functions 
      $\former(1^{|w|})$ and $\encoder(w)$
        (or respectively the single function $\semiencoder(w)$)
      for a (semi-)uniform family
        are computable in time polynomial in $|w|$ on a Turing machine we say the
      family uses \emph{polynomial time (semi-)uniformity}.
    If the uniformity functions are computable 
      by $\DLOGTIME$ uniform constant depth circuits,
        that is, $f,e,h \in \FAC^0$, 
          then the family is said to use  constant depth uniformity. 

    In this paper we consider two classes of problems, 
      those that can be solved by 
      $\FAC^0$-uniform families of confluent $\AMowoDiss$ 
        (active membranes without charges or dissolution rules)
        that run in time polynomial in $|w|$, denoted
        $\UniformMembrane$, 
      and 
      $\FAC^0$-semi-uniform families of confluent $\AMowoDiss$ systems 
        that run in time polynomial in $|w|$, denoted
        $\SemiUniMembrane$. 

    \subsection{Context-freeness in membrane systems}

    \begin{lemma}\label{lem:one_step}%
    Let $o$ be an object 
       in a membrane with label $h$ 
       in a configuration $\membraneConfig_i$  
       of a membrane system~$\Pi$. 
    Remove all other objects from $\membraneConfig_i$ to get configuration  $\membraneConfig^{\emptyset}_i$.
    If there is a rule $r$ in $\Pi$ 
      such that by applying that rule to $o,h$ in $\membraneConfig^{\emptyset}_i$ gives 
        a configuration $\membraneConfig^{\emptyset}_{i+1}$ with object $o'$ in $h'$,
      then
        it is the case that from configuration $\membraneConfig_i$
          there exists a configuration $\membraneConfig_{i+1}$ reachable in a single step that contains~$o'$ in~$h'$.
  \end{lemma}

  \begin{proof}
    The rule $r$ is of the type ($\arule$), ($\brule$), ($\crule$) or ($\erule$) as described in Definition~\ref{def:membrane}.
    It is sufficient to show that 
      there is always at least one maximal set of rule applications for
      configuration $\membraneConfig_i$ that creates $o'$ in $h'$ in $\membraneConfig_{i+1}$.
    
    Recall that an object in a configuration can be involved in at most one rule of any type.
    If the rule~$r$ is of type ($\arule$),
      it has the form $\EvolveRule{o}{h}{o' w}$ 
        where $w$ is a (possibly empty) string over~$\setOfObjects$ 
        and it is necessarily the case that $h = h'$ (rules of type ($\arule$) are applied within a single membrane).

    Let the notation $\membraneConfig_i - \{ o \}$ denote 
      the configuration $\membraneConfig_i$ 
      without the instance of the object $o$ under consideration
      and consider any maximal multiset $R$ of rules that can be applied 
      to the configuration $\membraneConfig_i - \{ o \}$. 
    Also, consider the multiset of rule applications $R$ unioned with the
    application of the rule $\EvolveRule{o}{h}{o' w}$ to our object instance~$o$
    in the relevant membrane with label $h$ in $\membraneConfig_i$.
    We claim that this new multiset is a maximal multiset of rules that can be applied
    to $\membraneConfig_i$.
    To see this notice that object instance $o$ has a rule being applied to it, 
    and each object can have at most one rule applied to it, 
    and no other objects with applicable rules are without rules because $R$ was maximal. 
    Hence there is a maximal multiset of rule applications for
    $\membraneConfig_i$ that applies $r$ 
    and hence when it is applied $\membraneConfig_{i+1}$ contains $o'$ in a membrane with label~$h=h'$.

    Rules of type ($\brule$), ($\crule$) and ($\erule$) involve both an object and a membrane.
    Consider $\membraneConfig_i - \{ o \}$  defined as above,
      and consider any maximal multiset $R$ of rule applications to $\membraneConfig_i - \{ o \}$.
    Furthermore,
      if in $R$ there is a rule that involves the membrane with label $h$ 
      where object instance $o$ was,
      then remove that rule application  from $R$ to get $\hat{R}$.
    We claim that
      the multiset of rule applications $\hat{R}$, unioned together with the rule application  
      ``rule $r$ applied to our object instance $o$ contained in the membrane with label~$h$'' is 
      a maximal multiset of rule applications for $\membraneConfig_i$.
      To see this note that 
        (i) $r$ is now being applied to the relevant instances of $o,h$ 
            so no other rule can be applied to that object nor to the membrane with that label,
        and 
        (ii) there are no other rules that can be applied  because~$R$ was maximal.
      After the application of this maximal multiset of rules 
        the new configuration~$\membraneConfig_{i+1}$ contains $o'$ in a membrane with label~$h'$.
  \end{proof}

  The following lemma generalises Lemma~\ref{lem:one_step} from one to multiple
  computation steps, and applies it to the setting of systems that recognise
  languages.   Intuitively, it states that if a sequence of rules $r$ can be
  applied starting from some configuration it is not possible to prevent this
  from happening by adding new objects to that configuration.
  \begin{lemma}  
    \label{lem:many_steps}
    Let $\Pi$ be a recogniser, \emphatic or acknowledger membrane system.
    Let $o$ be an object in a membrane with label $h$ in a configuration $\membraneConfig_i$ of $\Pi$. 
    Remove all other objects from $\membraneConfig_i$ 
      to get configuration $\membraneConfig^{\emptyset}_i$.
    If starting from configuration $\membraneConfig^{\emptyset}_i$ 
      there is a computation that halts after $t$ steps on configuration
      $\membraneConfig_{i+t}$ that contains object $\yes$ in the environment 
    then it is the case that starting from configuration $\membraneConfig_i$
      there exists a halting computation with $\yes$ in the environment.
  \end{lemma}
  \begin{proof}
  By hypothesis we know that there is a sequence of $t$ rules $r_1,r_2,\ldots, r_t$
    that can be applied to  $\membraneConfig^{\emptyset}_i$ to get $\yes$ in the environment. 
  We apply Lemma~\ref{lem:one_step} $t$ times,
    first to configuration $\membraneConfig_i$ with $r = r_1$,
    then to $\membraneConfig_{i+1}$ with $r = r_2$,
    and so on until we get configuration $\membraneConfig_{i+t}$ which contains $\yes$ in the environment.  

  If $\Pi$ is a recogniser system then we are done:
    recogniser systems produce $\yes$ \emph{in the halting step}. %
  If $\Pi$ is a  \emphatic or acknowledger membrane system we add the fact  
    (from Section~\ref{sec:recog}) that no rules can be applied to the object $\yes$,
    and since there is a computation where $\yes$ is in the environment at  
    configuration $\membraneConfig_{i+t}$,
    then it remains there until the computation eventually halts.
  \end{proof}
  
  Lemma~\ref{lem:many_steps} shows that the kind of membrane systems studied in
  this paper intuitively exhibit some notion of context-freeness. Essentially, there is a sense in which an object $o_s$ can be said to trigger a sequence of rules that eventually result in the production of object $o_t$ on some computation, and specifically, the production of $o_t$ can not be prevented by starting over with more objects (more context) in the system. Hence, the ideas used in the proof of
   Lemma~\ref{lem:many_steps} justify the use of the following definition in our proofs. 
  
    \begin{definition}[Eventually evolves]
    \label{def:eventually_evolves}
    Let~$\membraneConfig_s$ be a configuration of a membrane system~$\Pi$, containing an object of type~$o_s$ in a membrane labelled~$h_s$ (along with any number of other objects and membranes). 
    Let~$\membraneConfig_s^{\emptyset}$ denote $\membraneConfig_s$ with all objects removed except one instance of $o_s$ in the relevant membrane with label $h$. 
    We say that $o_s$ in~$h_s$ in configuration $\membraneConfig_s$   
    \emph{eventually evolves on some computation path},
      or for short \emph{eventually evolves},
     object type~$o_t$ in a membrane labelled~$h_t$
    if there is a computation (sequence of configurations) starting from 
    $\membraneConfig_s^{\emptyset}$
    where the final configuration in the computation has  
    object type~$o_t$ in a membrane labelled~$h_t$.
  \end{definition}
  Note that if $o_s$ in $h_s$ in $\membraneConfig_s$ eventually evolves $\yes$ in $\env$
    this means that by Lemma~\ref{lem:many_steps} 
    there is at least one computation (sequence of configurations) that leads to a configuration with $\yes$ in $\env$ from~$\membraneConfig_s$. 
  However, since membrane systems are nondeterministic, this does not necessarily happen for all computations. 

  \section{Uniformity is strictly weaker than semi-uniformity}
  \label{sec:uni_noteq_semi} 
  Theorem~\ref{lem:uni_less_semi} proves that uniform families of membrane systems 
  are strictly weaker than semi-uniform families of the
  same type. The result holds for all three definitions of
  acknowledger, \emphatic and recogniser membrane systems.
  
  \begin{lemma}
    \label{lem:uni_in_nonuniAC}
    $\UniformMembrane \subseteq \text{non-uniform-}\AC^0$,
      for acknowledger, \emphatic and recogniser membrane systems.
  \end{lemma}

  \begin{proof}
    Let $L \in \UniformMembrane$, 
      and let
      $\membraneFamily$ be the $\FAC^0$-uniform family of that type that decides $L$. 
    That is, given $w \in \{0,1\}^\ast$ there is a membrane system 
      $\Pi_{|w|} \in \membraneFamily$ that accepts $\encoder(w)$ iff $w \in L$. 

    We now describe a non-uniform family of constant-depth circuits
    ${\mathfrak{C}} = \{ C_n \mid n \in \mathbb{N}$ $ \text{ and } C_n \text{
      accepts }$ $  L \ \cap \{ 0,1\}^n \}$
      that recognizes $L$.
    For any input $w\in \{ 0,1\}^*$, we claim that circuit $C_{|w|} \in
      {\mathfrak{C}}$ decides whether or not $w \in L$.
    The first constant number of layers of the circuit $C_{|w|}$ compute
      the input encoding function $\encoder(w) \in \FAC^0$. 
    This generates a polynomial (in $|w|$) number of binary words that encode
      elements from the polynomially sized object set $\setOfObjects$  as well as their multiplicities (in unary).
      
    The  circuit $C_{|w|}$ then converts 
      the list of encoded $\encoder(w)$ objects into a
      single binary string $\chi$ of length~$|\chi| = |O|$ 
        such that for all $i\in \{1,2,\ldots ,|\setOfObjects| \}$, 
        the $i$th bit $\chi_i = 1$ iff 
        $o_i \in \setOfObjects$ is in $ \encoder(w)$. 
   That is, $\chi$ is a characteristic sequence for $\encoder(w)$, ignoring   multiplicities. 
    
    For each $i$, the bit $\chi_i$ is wired into a unique $\ANDgate$ gate $a_i$,
      giving a total of $|O|$ $\ANDgate$ gates at this level. 
    The second input to the $\ANDgate$ gate $a_i$ is from a constant gate $c_i$,
    where  $c_i = 1$ if $o_i \in \setOfObjects$ in the input
      membrane eventually evolves (Definition~\ref{def:eventually_evolves}) to the $\yes$ object in the $\env$ membrane and $c_i = 0$ otherwise.

    The next layer contains a single $\ORgate$ gate $g$ such that  
    for each $i$, $\ANDgate$ gate $a_i$ is wired to $g$. 
    This $\ORgate$ gate  is the output gate of the circuit. 
    Also wired into the $\ORgate$ are $|\setOfObjects| \times |\setOfLabels|$
      constant gates such that gate $c_{o,h} = 1$
        if both 
          (i) $o \in \setOfObjects$ is in membrane labelled $h \in \setOfLabels$ 
            in the initial configuration of $\Pi_{|x|}$ 
         and (ii) $o$ in $h$ eventually evolves to $\yes$ in the $\env$ membrane,
        otherwise $c_{o,h}= 0$.
        
    We now argue that the above construction of $C_{|w|}$ accepts $w\in L$. 
    Recall that $\Pi_{|w|}(\encoder(w))$ is a confluent membrane system and 
    so if the computation is an accepting one, then all possible computation
    paths are accepting. 
    For a computation to be accepting, a $\yes$ object must appear in the
    $\env$ membrane. 
    Therefore at least one object in the initial configuration of 
      $\Pi_{|w|}(\encoder(w))$ must eventually evolve
       to be a $\yes$ in the
      $\env$ membrane. %
    Also $\Pi_{|w|}(\encoder(w))$ is confluent, therefore if at least one object in the initial configuration of 
      $\Pi_{|w|}(\encoder(w))$ eventually evolves $\yes$ in the
      $\env$ membrane, the system accepts. 
    Since the property of whether an object in some membrane eventually
    evolves to object $\yes$ in the $\env$ membrane depends only on 
      $R$ and $\mu$ in $\Pi_{|w|}$, and hence in turn  
      depends only on~$|w|$ (by Lemma~\ref{lem:many_steps}),  
      it can be 
      encoded (non-uniformly) in the constants $c_i$ in circuit $C_{|w|}$. 
     
    Suppose $\Pi_{|w|}$ accepts regardless of the input $\encoder(x)$. 
    In this case one of the objects, say $o$, in the initial configuration of 
    $\Pi_{|w|}$ will eventually evolve to $\yes$ in the $\env$ membrane.
    This means the relevant gate~$c_{o,h}$ will be a 1-constant gate and so the output
    $\ORgate$ will evaluate to 1 and so $C_{|w|}$ accepts regardless of input.

    Suppose $w \in L$, therefore 
      at least one of the objects in 
      in $\encoder(w)$, when placed in the input membrane of
      $\Pi_{|w|}$, yields a computation that ends with a configuration with 
       object $\yes$ in membrane $\env$. 
    In turn this implies that  
      at least one of the $\ANDgate$ gates $a_i$ 
        has inputs  $c_i = 1$ 
        and  $\chi_i = 1$ and so evaluates to 1. 
    Finally this causes the $\ORgate$ to evaluate to 1 and so $C_{|w|}$ 
    accepts input $w$. 
    
    Suppose $w \notin L$, in this case none of the objects in 
      $\encoder(w)$ will
        eventually evolve to $\yes$ in the $\env$ membrane.
    Thus any of the $a_i$ $\ANDgate$ gates that have
      a constant $c_i=1$ as input 
      will  have  $\chi_i=0$ and so will evaluate to 
      0. 
    With all 0 inputs, the output $\ORgate$ evaluates to 0 and the circuit
      rejects.  

    This circuit is of polynomial size and its depth is the sum of the 
      depths of the $\FAC^0$ encoding function (which has depth $O(1)$, by definition),
      the depth of the circuit that converts $\encoder(w)$ into $\chi$
      (which is $O(1)$ using masking and comparison), 
      and  2 for the final layer of  $\ANDgate$ gates and the single $\ORgate$ gate.
    Hence ${\mathfrak{C}}$ is a non-uniform-$\AC^0$ circuit family that recognizes $L$.
  \end{proof}

  \begin{theorem}
    \label{lem:uni_less_semi}
    $\UniformMembrane \subsetneq \SemiUniMembrane$, for acknowledger, \emphatic and recogniser membrane systems. %
  \end{theorem} 
  \begin{proof}
    ($\subseteq$)
    By definition, uniform families are a restriction of semi-uniform families 
    and so\linebreak
      $\UniformMembrane \subseteq \SemiUniMembrane$.
      
    \noindent
    ($\neq$) 
    $\Parity \subseteq \{ 0,1\}^*$ is the set of binary strings that
      contain an odd number of 1s. 
    We claim that $\Parity \in \SemiUniMembrane$ for recogniser systems  (and
      hence also for acknowledger and \emphatic membrane systems).
    Let $w \in \{0,1\}$, $n = |w|$, and let $w = w_1, \ldots, w_n$.
    We will define the function $\semiencoder\colon \{0,1\}\mapsto\membraneFamily$, where  
      each $\semiencoder(w) = \Pi_w$ computes 
      $\chi_\Parity(w)$ as follows. 
    Each~$\Pi_w$ has a single membrane, $\env$, 
      the set $\setOfObjects$ contains 
            $2n + 2$ objects: $\setOfObjects = \{ o_i |  1 \leq i \leq n\} \cup \{ e_i |  1 \leq i \leq n \} \cup \{ \yes, \no \}$. 
    The initial configuration is 
      the membrane $\env$ containing a single
      object $o_1$ in $\env$ if $w_1 = 1$
      or object $e_1$ in $\env$ if $w_1 = 0$.
    The rules of $\Pi_w$ are as follows: 
      if $w_i = 1$ then
        $\EvolveRule{o_i}{\env}{e_{i+1}}$,  
        $\EvolveRule{e_i}{\env}{o_{i+1}}$ and 
      if $w_i = 0$ then
        $\EvolveRule{e_i}{\env}{e_{i+1}}$,  
        $\EvolveRule{o_i}{\env}{o_{i+1}}$. 
    There are also the rules
      $\EvolveRule{e_n}{\env}{\no}$ and 
      $\EvolveRule{o_n}{\env}{\yes}$. 

    By starting with object $o_1$ if $w_1 = 1$,  and 
    then changing between $e_i$ and $o_{i}$ if $w_i=1$, and not changing  if $w_i =0 $ at each timestep
    then we ensure that the object $o_i$ represents ``the parity of the first $i$
    bits of $w$ is odd'', and $e_i$ represents that they are even.
    Thus, $o_n$ evolves to a single $\yes$ object if there is an odd number 
    of 1s in $w$ and $e_n$ evolves to a single $\no$ if there is an even
    number of 1s in $w$. 
    
    To end we note that it is known~\cite{FSS1984p} that  
      $\Parity \notin \textrm{non-uniform-}\AC^0$. 
    Lemma~\ref{lem:uni_in_nonuniAC}  shows that\linebreak 
      $\UniformMembrane~\subseteq~\text{non-uniform-}\AC^0$,
      for acknowledger, \emphatic, and recogniser membrane systems. 
  \end{proof}

  \section{The computational power of semi-uniform families}
  \label{sec:semi_uniform}
  In prior work~\cite{MW2011}, we  have shown that   semi-uniform families of 
  \emphatic 
  membrane systems %
  characterise $\NL$.
  We  give an alternative proof here to demonstrate 
  techniques that we 
  will use in later sections for uniform families.

  \begin{theorem}[\cite{MW2011}]
    \label{thm:PMCwodiss-is-NL}
    $\ACPMCwodissSemi = \NL$,   for both acknowledger and \emphatic membrane systems. 
  \end{theorem}
  \begin{proof}
  Lemmas~\ref{lem:PMCwodiss-sub-NL} and~\ref{lem:NL-sub-PMCwodiss} give the
  proof for acknowledger and \emphatic membrane systems.
  \end{proof}

  \begin{lemma}[\cite{MW2011}]
    \label{lem:PMCwodiss-sub-NL}
    $\ACPMCwodissSemi \subseteq \NL$, for acknowledger, \emphatic and recogniser membrane systems.
  \end{lemma}
    \begin{proof}
      Let~$\membraneFamily$ be a semi-uniform family 
        of acknowledger, \emphatic or recogniser membrane systems 
        that recognises~$L \in \ACPMCwodissSemi$. 
      Let $\semiencoder\colon \{0,1\}^* \mapsto \membraneFamily$ 
        be the semi-uniformity function of $\membraneFamily$, 
        that is, on input $x \in \{0,1\}^*$, $\Pi_x = \semiencoder(x)$ accepts iff
        $x \in L$.
      We present a non-deterministic logspace Turing
        machine~$M$ that recognises~$L$.

      The computation of $M$ proceeds as follows:
      First $M$, on input $x$, non-deterministically chooses a single object from 
        $\membraneConfig_1$, the initial configuration of $\Pi_x$, and stores 
        (a string representation of) the object and its containing membrane on its work tape. 
      Then $M$ enters a loop where at each iteration it
        non-deterministically
        chooses one of the rules applicable to the object on its work tape. 
      If the rule is of type~($\arule$) or~($\erule$) (Definition~\ref{def:membrane}) then 
        $M$ replaces the current object on the work tape (the membrane remains
        unchanged)
        with a non-deterministically chosen object from  
        the right hand side of the rule. 
      If the rule is of type ($\brule$) or ($\crule$) then
        the object on the work tape
        is replaced by the object on the right hand side of the rule 
        and the membrane on the work tape is replaced by 
         the parent (type (\crule)) or child membrane (type (\brule))
         of the current membrane. %
      If during the computation the work tape is found to store the
        object $\yes$ in the $\env$ membrane then 
        $M(x)$ halts and accepts. 
      Otherwise, if there are no rules applicable to the object and membrane on the 
        work tape, and it is not $\yes$ in $\env$, then $M(x)$ halts and rejects. 

      Suppose that $x \in L$ and so $\Pi_x = \semiencoder(x)$ accepts. 
      This implies that there is one (or more) 
        objects in the initial configuration of $\Pi_x$ that will, 
        by the application of rules to this object and its successors, 
        become the object $\yes$ in the $\env$ membrane 
        by the end of the computation of $\Pi_x$ 
        (this claim follows from the kind of rules we allow---they are essentially context free---and can be formally proven using dependency graphs~\cite{NJNC2006p}).
      Indeed, this observation holds for all three kinds of membrane systems:  acknowledger, \emphatic and recogniser. 
      By non-deterministically choosing an object in the initial 
        configuration, and non-deterministically choosing the 
        rules that are applied to this object and its successors we ensure that 
        there is a computation of $M(x)$ for each possible 
          sequence of rule applications of $\Pi_x$ for 
          each object in the initial configuration $\Pi_x$ 
          (this follows from Lemma~\ref{lem:many_steps}). 
      Therefore at least one computation of $M(x)$ will produce the object
        $\yes$ in the $\env$ membrane and so $M(x)$ accepts, by confluence. 
      That is, if $\Pi_x$ accepts then $M$ accepts on input~$x$. 

      Suppose that $x \notin L$ and so $\Pi_x = \semiencoder(x)$ rejects.
      This implies that 
          there is no valid computation of~$\Pi_x$ 
          where an object in the initial configuration evolves to $\yes$ in the
          $\env$ membrane.
                  Indeed, this observation holds for all three kinds of membrane systems:  acknowledger, \emphatic and recogniser. 
      In this case \emph{all} computation branches of $M(x)$ will reach an
        object to which no further rules are 
        applicable (that is not $\yes$) and so will halt in the rejecting state. 
      That is, if $\Pi_x$ rejects then $M$ rejects on input $x$.

      To simulate the computation of $\Pi_x$ in logspace, $M(x)$ 
        recomputes relevant logarithmic sized pieces of $\semiencoder(x) = \Pi_x$ via the classic technique for composing logspace algorithms  (see Chapter 4.3 of~\cite{AB2009x}) each time it needs information about 
        $\Pi_x$, i.e.\ initial configuration, rules, or membrane structure. 
      From the statement, $\semiencoder$ is computable in $\FAC^0$. 
        This means that the number of unique objects and labels in $\Pi_x$
        are polynomial in $n = |x|$ and so each can be
        uniquely identified in binary with a string of length $\log n$. 
     $M(x)$ uses a constant number of $\log n$ sized binary strings 
        to encode the current object and membrane, as well as 
        some counters and temporary storage needed to re-compute
        $\semiencoder(x)$. 
      
      Therefore $L$ is  decided by a
      non-deterministic logspace Turing machine.
  \end{proof}

      \begin{lemma}[\cite{MW2011}]
        \label{lem:NL-sub-PMCwodiss}
         $\NL \subseteq \ACPMCwodissSemi$, for acknowledger and \emphatic membrane systems.
      \end{lemma}
      \begin{proof}
        Let $L \in \NL$.
        That is,
          there is a 
            non-deterministic logspace Turing machine~$M$
            with one or more accepting computation paths exactly
              for input words $x \in L \subseteq \{ 0, 1 \}^*$.
        
      We show that there is 
        an $\FAC^0$ semi-uniform  family of polynomial-time membrane systems 
        $\membraneFamily$ that recognises $L$.
      We now describe a function $\semiencoder\colon \{0,1\}^* \rightarrow \membraneFamily$,
          computable in $\FAC^0$,
            such that 
              if $x \in L$ then
                $\semiencoder(x) = \Pi_x  $ accepts,
              otherwise 
                $\Pi_x$ rejects. 

      Consider the configuration graph $G_{M,x}$ for~$M$ on input $x\in\{0,1\}^*$,
          which is $\FAC^0$ computable from~$M$ and $x$ (see
          Section~\ref{sec:config_graphs} and Lemma~\ref{lem:config_in_AC0}).
      Also consider the Turing machine $N_M$ 
        (and its configuration graph $G_{N,x}$)
          that on input $x$ 
          accepts only if all computations of $M$ reject on input $x$, 
          that is, $x \notin L$.
      $N_M$ uses the standard 
            un-reachability algorithm~\cite{Imm1988p,Sze1988p} for
            non-deterministic logspace.

        The function $\semiencoder(x)$ constructs the configuration graph
          $G_{M,x}$ and modifies it to produce a membrane system
          $\Pi_{x}$ as follows. 
        $\setOfObjects$, the set of unique objects of $\Pi_{x}$ 
          has an object encoding each vertex in the configuration
          graphs $G_{M,x}$ and $G_{N,x}$ as well as two extra objects,
          $\yes$ and $\no$.
        The initial configuration of $\Pi_{x}$ has a single membrane
          labelled $\env$ that contains two objects:
            $c_i$ which encodes the initial configuration of $M(x)$; and 
            $c_j$ which encodes the initial configuration of $N_M(x)$. 
        The edges of the configuration graphs $G_{M,x}$ and $G_{N,x}$ 
          are encoded as object 
          rewriting rules in the membrane system. 
        If vertex $u$ has $k$ edges to vertices  $v_1, \ldots, v_k$ then 
          $\semiencoder(x)$ 
        encodes all $k$ edges as a single type~($a$) rule:
        $\EvolveRule{u}{\env}{v_1,\ldots,v_k}$.
        Let vertex (object)~$c_a$ encode the accepting configuration of the
          Turing machine~$M$, and let $\semiencoder(x)$ include the rule
          $\EvolveRule{c_a}{\env}{\yes}$.
        Likewise for the vertex (object) $c_b$ that encodes an accepting
          configuration of the Turing machine $N_M$, $\semiencoder(x)$ includes
          the rule $\EvolveRule{c_b}{\env}{\no}$.

      We now argue that each member
        $\Pi_{x} = \semiencoder(x)$
        of the semi-uniform family $\membraneFamily$,
        accepts iff $x \in L$. 

      Suppose $x \in L$,
        therefore Turing machine $M(x)$ accepts.
      This implies that 
        configuration graph $G_{M,x}$  has the property that there
        is a
        directed path from the vertex $c_i$ representing the initial 
        configuration, to the accept vertex $c_a$.
      The assumption also implies that $N_M(x)$ must reject, 
        and so configuration graph $G_{N,x}$ does not 
        have a directed path from 
          the object $c_j$ encoding its initial configuration 
          to $c_b$, its accept configuration. 
      Since $\Pi_x = \semiencoder(x)$ directly encodes
        the configuration graphs as objects and rules 
        then the existence of a path from $c_i$ to $c_a$ implies that 
        the membrane system will produce the object $\yes$ during its
        computation. 
      The absence of a path from $c_j$ to $c_b$ implies that 
        the membrane system will not produce the object $\no$ during its
        computation. 
      Therefore $\Pi_x$ accepts if $x \in L$.

      Suppose $x \notin L$,
        therefore no computation paths of Turing machine $M(x)$ accept.
      This implies that 
        configuration graph $G_{M,x}$  has the property that there
        is no 
        directed path from the vertex $c_i$ representing the initial
        configuration, to the accept vertex $c_a$.
      The assumption also implies that $N_M(x)$ must accept, 
        and so configuration graph $G_{N,x}$ 
        has a directed path from 
          the object $c_j$ encoding its initial configuration 
          to $c_b$, its accept configuration. 
      Since $\Pi_x = \semiencoder(x)$ directly encodes
        the configuration graphs as objects and rules 
        then the existence of a path from $c_j$ to $c_b$ implies that 
        the membrane system will produce the object $\no$ during its
        computation. 
      The absence of a path from $c_i$ to $c_a$ implies that 
        the membrane system will not produce the object $\yes$ during its
        computation. 
      Therefore $\Pi_x$ rejects if $x \notin L$.

      Since each configuration graph is acyclic and has $p(|x|)$ nodes where 
        $p$ is some polynomial function,
        it follows that the membrane system itself is of polynomial
          size and halts in polynomial time.
      The configuration graph can be computed in $\FAC^0$ by
      Lemma~\ref{lem:config_in_AC0}.
        
      In conclusion, function $\semiencoder$
        defines a semi-uniform family of polynomial time
        $\AMowoDiss$ \emphatic
          (and so also acknowledger)
            membrane systems that accept the language in $L \in \NL$.
    \end{proof}

    Note that the above proof fails for recogniser membrane systems since
    if there is more than one accepting computation (or in the rejecting case, 
      more than one rejecting computation) then multiple copies of the object 
    $\yes$ (or $\no$) are produced in violation of the definition of
    recogniser membrane systems. 
    
  \section{The computational  power of uniform families of acknowledger membrane systems}
  \label{sec:membrane-uni}
  
  In this section we focus on acknowledger membrane systems (Definition~\ref{def:acknowledger}) where 
  the accepting condition is met by the presence of one or more  $\yes$ object in the environment in the last step of a computation, and the absence of $\yes$  implies rejection. 
  We give a characterisation of uniform families of 
  acknowledger membrane systems: %

  \begin{theorem}
    \label{thm:uni_is_dtt_tallyNL}
    $\UniformMembrane = \Reducible{dtt}{\FAC^0}{\tally\NL}$, for acknowledger membrane systems.
  \end{theorem}

  The proof of this result is the combination of 
    Lemmas~\ref{lem:unimem_subset_dttTallyNL}
    and~\ref{lem:dttTallyNL_subset_uniMem}. 
  Before giving the lemmas we first introduce the following $\FAC^0$
  computable functions that will be used in the proofs. 

    \paragraph{Pairing function}
    We require an injective function that pairs two binary strings into one and is extremely easy
    ($\FAC^0$) to compute.  We use the pairing function
    that interleaves the bits of two binary string arguments $a$ and $b$. For
    example, the binary strings $a=a_2 a_1 a_0$ and $b=b_2 b_1 b_0$ are paired
    as the interleaved string $\langle a,b \rangle =
    b_2 a_2 b_1 a_1 b_0 a_0$.
    The circuits for interleaving and de-interleaving have only a single input gate layer and
    a single output gate layer (and so have 2 layers). The wiring between
    each input and output gate can be shown to be
    $\DLOGTIME$-uniform. 

    \paragraph{Binary to Unary}
    There is a 
      constant depth
        circuit family where circuit $C_n$ 
        takes as input some word  $w\in \{0,1\}^n$ 
        and outputs $1^x$ where $x$ is the positive integer encoded in the 
        first $\lceil\log_2 n \rceil$ bits of $w$~\cite{CSV1984}. 
    It can be shown that this circuit family is  
      $\DLOGTIME$ uniform and so this conversion from short binary strings to unary is in $\FAC^0$. 

    \paragraph{Unary to Binary}
    There is a 
      constant depth
        circuit family where circuit $C_n$ 
        takes as input some word $w = 0^{n-x} 1^x$ where $0 \leq x \leq n $,
        and outputs the binary encoding of $x$~\cite{CSV1984}. 
    It can be shown that this circuit family is $\DLOGTIME$ uniform and so 
      unary to binary conversion is in $\FAC^0$.

  \begin{lemma} 
     $\UniformMembrane \subseteq
     \Reducible{dtt}{\FAC^0}{\tally\NL}$, for acknowledger, \emphatic and recogniser membrane systems.
     \label{lem:unimem_subset_dttTallyNL}
  \end{lemma}
  \begin{proof}
    Let $L \in \UniformMembrane$.
    That is, there exist two functions
        $\encoder , \former \in \FAC^0$, such that~$\encoder$ maps
          $x \in \{ 0,1 \}^*$ to a multiset of membrane system objects (the input),
        and $\former$ maps 
          $u \in \{ 1 \}^*$ to a membrane system, 
          $\former(1^{|x|}) = \Pi_{|x|} \in \membraneFamily$, such that $\Pi_{|x|}$ accepts  input $\encoder(x)$
          iff $x \in L$. 

    We claim that $L$ is $\FAC^0$ disjunctive reducible to 
          a unary language $T$, 
            where $T$  is decided by a non-deterministic logspace Turing machine
            $\tallyMachine$.

    Let $T$ be the set of words of the form $1^{\langle o, |x|\rangle}$ where, for all $|x| \in \mathbb{N}$ and then for all $o \in O_{ |x| } $,  
      membrane system $\Pi_{|x|} \in \membraneFamily$ accepts if
      object $o$ in the input membrane eventually evolves to
        the object~$\yes$ in the $\env$ membrane (where $O_{ |x| }$ is the set of objects of $\Pi_{|x|}$,  where both $x$ and $o$ are encoded in binary, 
 $\langle \cdot , \cdot \rangle$ is the binary interleaving function
      defined at the start of Section~\ref{sec:membrane-uni}, and   $1^{b}$ denotes the unary word over $\{1\}$ of length $b$ for a binary number $b$).
    Turing machine $\tallyMachine$ decides words in $T$ 
      by first converting the input word to binary and then reversing the
        pairing function to find $o_i$ and $|x|$. 
    $\tallyMachine$ then proceeds by simulating $\Pi_{|x|}$ in
    non-deterministic logspace using a similar method 
        as described in Lemma~\ref{lem:PMCwodiss-sub-NL}, that is, by 
        storing a constant number of objects and membranes on its work tape
        and recomputing $\former(1^{|x|})$ as needed (the main difference is that~$\tallyMachine$ uses $o_i$ as its starting
      object instead of non-deterministically choosing one).
    As in 
      Lemma~\ref{lem:PMCwodiss-sub-NL}, $\tallyMachine$ accepts 
      if there exists a valid computation in $\Pi_{|x|}$ where $o_i$
      in the input membrane becomes $\yes$ in the $\env$ membrane. 
    $\tallyMachine$ rejects if there are no valid computations that lead to 
      a $\yes$ object in the $\env$ membrane.
    Therefore $T$ is a tally language decided by a non-deterministic logspace
    Turing machine and so  $T \in \tally\NL$. 

    We now define the function $\tau \in \FAC^0$, 
      that maps from $\{0,1\}^*$ to the set of tuples of unary words,
          and later prove that 
             if $x \in L$ then $\tau(x)  \cap T \neq \emptyset$, 
             otherwise if $x \notin L$ then $\tau(x) \cap T = \emptyset$.
    Let $\tau(x) = (u_1,\ldots, u_{q(|x|)})$,
      where $q(|x|)$ is the number of object types $o_i$ in $\encoder(x)$,  and $u_i = 1^{\langle o_i, |x| \rangle}$. 
    Note that the set of unique words in $\tau(x)$ is a bijection onto the set of 
            objects $\encoder(x)$ so $q(|x|)$ is polynomial in~$|x|$.
    Since $\encoder$, the pairing function, binary-unary
      conversions, as well as calculation of $q(|x|)$  are in $\FAC^0$, 
      it is not difficult to see that $\tau \in \FAC^0$.
    
    We now prove that $\tau$ is a disjunctive reduction from $L$ to $T$.   
    Suppose $x \in L$,
        this implies that \emph{at least one} of the objects in $\encoder(x)$, 
        when placed in the input membrane of $\Pi_{|x|}$
        evolves to a $\yes$ object in the $\env$ membrane by the end of the
        computation of $\Pi_{|x|}$. 
    Then, by the definition of $\tau$, 
      if $x \in L$
        then $\exists o \in \tau(x)$ such that $o \in T$. 
    
    Let $x \notin L$,
        this implies \emph{none} of the objects in $\encoder(x)$, 
        when placed in the input membrane of $\Pi_{|x|}$, 
        evolve to a 
        $\yes$ object in the $\env$ membrane by the end of the computation of $\Pi_{|x|}$. 
    Then, by the definition of $\tau$, 
      if $x \notin L$
        then $\nexists o \in \tau(x)$ such that $o \notin T$.
\end{proof}

  \begin{lemma}
     $\Reducible{dtt}{\FAC^0}{\tally\NL}
     \subseteq
     \UniformMembrane$, for acknowledger membrane systems.
     \label{lem:dttTallyNL_subset_uniMem}
  \end{lemma}
  \begin{proof}
    Let $L \in \Reducible{dtt}{\FAC^0}{\tally\NL}$.
    That is,
      there exists a unary language $T \subseteq \{1\}^*$ 
        that is recognised by 
          a non-deterministic logspace Turing machine
            $\tallyMachine$,
      and a function $r \in \FAC^0$ that maps
        $x \in \{ 0,1 \}^*$ to 
          a set of unary words such that $r(x) \cap T \neq \emptyset$
          if $x \in L$, and $r(x) \cap T = \emptyset$ otherwise. 
    Let $q'(|x|)  = \max(\{ \max(|r(w)|) \mid w \in \{0,1\}^{|x|} \})$, 
      that is, the length of largest word produced by $r$ on any input of
      length $|x|$. 
    Note that $q'(|x|)$ is computable by $\former$ since $r \in \FAC^0$. 

    We present an $\FAC^0$ uniform polynomial-time $\AMowoDiss$ membrane family
      $\membraneFamily$ that recognises $L$.
    The family is composed of two functions:
      the uniformity function  
        $\former\colon \{1\}^* \rightarrow \mathbf{\Pi}$;
      and $\encoder$ that maps from binary words to the multiset of unique objects in 
      the appropriate member of $\membraneFamily$.

    Each member $\Pi_{|x|} = \former(1^{|x|})$ of $\membraneFamily$ has one
    single membrane, $\env$, that is both the environment and the input
    membrane.
    On input  $1^{|x|}$, 
      the function 
    $\former$ produces a 
      configuration graph 
        $G_{\tallyMachine,u}$ for  
            machine~$\tallyMachine$ on input $1^u$
              for each $u \in \{ 1,2 ,\ldots, q'(|x|) \}$.
    (Note that this is a generalization of the technique used in the proof of
        Lemma~\ref{lem:NL-sub-PMCwodiss}.)
    Since we have unary input words we can include the input word as part of
    the configuration to ensure that there is a unique input configuration for 
    each $G_{\tallyMachine,u}$.
   
    Each of the $q'(|x|)$ configuration graphs are
      converted to membrane rules and objects,
      using the same technique (without the second Turing
      machine that solves un-reachability) from the proof of
        Lemma~\ref{lem:NL-sub-PMCwodiss}, 
          of a single membrane system $\Pi_{|x|}$. 
    In summary, the vertices of the configuration graphs become objects in~$\Pi_{|x|}$ 
      and the edges in the graph become type ($a$) rules.
    There is a type (a) rule that maps the object encoding the accepting configuration of 
      $\tallyMachine$ to $\yes$. 
    We do not include the second Turing machine that solves un-reachability from 
        Lemma~\ref{lem:NL-sub-PMCwodiss}.
    $\tallyMachine$ is a logspace machine and 
    so its configuration graph is of polynomial size, it follows that the membrane 
    system is of polynomial size.
    It is relatively straightforward to verify that $\former \in \FAC^0$.

     The input encoder $\encoder(x)$ simulates 
        $r(x)$ to find the set of unary words $(u_1, \ldots, u_k)$,
        then outputs an object $c_{i,u}$ for each $u \in r(x)$, %
        which encode the vertex of the configuration graph corresponding to
        the initial configurations of Turing machine $\tallyMachine$ 
        input $u$. 
    Since $r \in \FAC^0$ it is not difficult to see that $\encoder \in \FAC^0$.

    We now show that 
      the membrane system $\Pi_{|x|}$ on input $\encoder(x)$ accepts if $x \in L$ and
      otherwise rejects. 

      Suppose $x \in L$. 
      This implies that at least one word in $r(x)$ is in the tally set $T$ 
        and so $\tallyMachine$ accepts on at least one of these inputs.
      The input membrane of $\Pi_{|x|}$ contains $\encoder(x)$ 
        which includes the object $c_{i,u}$ which encodes the configuration graph vertex that represents 
        the initial configuration of Turing machine~$\tallyMachine$ on input
        $1^u$.
       In the proof of Lemma~\ref{lem:NL-sub-PMCwodiss} we show 
         how the construction of $\Pi_{|x|}$ is such that 
         there is a sequence of rules from the input object $c_{i,u}$ to the 
         $\yes$ object and so 
          $\Pi_{|x|}$ on input $\encoder(x)$ will accept. 
      Suppose $u \notin L$.
      This implies that none of the unary words  $r(x)$ are in the tally set $T$ 
        and that $\tallyMachine$ does not have any accepting computations on 
        any of the words $1^j$ in $r(x)$.
      So, as  in the proof of Lemma~\ref{lem:NL-sub-PMCwodiss},        this 
        implies that none of the objects in $\encoder(x)$ in the input
        membrane of $\Pi_{|x|}$ can evolve to  
          the object $\yes$ in the $\env$ membrane.
      In this case the membrane system $\Pi_{|x|}$ on input $\encoder(x)$
        will halt without $\yes$ object; a rejecting computation for an acknowledger membrane system.
        
      Therefore the pair of functions $\former$
        and $\encoder$ provide a uniform family of polynomial time 
        $\AMowoDiss$ membrane systems that accept 
        $L \in \Reducible{dtt}{\FAC^0}{\tally\NL}$.
  \end{proof}

  \section{The computational  power of \emphatic membrane systems}
  \label{sec:stricter}
  In this section we further investigate how the details in the definition of acceptance and rejection for recogniser membrane systems affect the computational power of uniform families of 
  $\AMowoDiss$ systems. 

  In Section~\ref{sec:membrane-uni} we consider acknowledger membrane systems (Definition~\ref{def:acknowledger}) where 
    the absence of a $\yes$ object in the environment in the last step of any computation of a membrane system is sufficient to say that the system rejected its
    input. 
  However, if we restrict to \emphatic membrane systems, which must produce 
  one or more $\yes$ objects in the case of an accepting computation
  and one or more $\no$ objects in the case of a rejecting computation (Definition~\ref{def:emphaticrecogniser})
    it is no longer clear if our characterisation 
     of $\UniformMembrane$ for acknowledger systems can still hold.
  The best lower-bound we find is
    $\Reducible{m}{\FAC^0}{\tally\NL}$, and we obtain upper-bounds of $\Reducible{dtt}{\FAC^0}{\tally\NL}$ and 
    $\Reducible{ctt}{\FAC^0}{\tally\NL}$.

  In the semi-uniform case the upperbound 
       $\ACPMCwodissSemi \subseteq \NL$ is unaffected by the restriction from acknowledger to \emphatic membrane systems. 
  It also turns out that these more restricted \emphatic membrane systems have
  the same $\NL$ lower-bound on their power as acknowledger membrane systems
  (see Lemma~\ref{lem:NL-sub-PMCwodiss}).

  \begin{lemma}
     $\UniformMembrane \subseteq
     \Reducible{ctt}{\FAC^0}{\tally\NL}$, for \emphatic and recogniser membrane systems.
     \label{lem:strict_unimem_in_ctt_tallynl}
  \end{lemma}
    \begin{proof}
    \emph{(Sketch)}
    This proof closely follows  that of Lemma~\ref{lem:unimem_subset_dttTallyNL}
      so we just highlight the differences. 
    In the proof of Lemma~\ref{lem:unimem_subset_dttTallyNL} the 
      language $T$ is 
        the set of words $1^{\langle o, |x|\rangle}$ where 
          membrane system $\former(1^{|x|}) = \Pi_{|x|}$ accepts if
          object~$o$ in the input membrane eventually evolves to
            the object $\yes$ in the $\env$ membrane.
    In this proof we consider the language $T'$ that is
      the set of words $1^{\langle o, |x|\rangle}$ where 
        in the membrane system $\Pi_{|x|} $ %
        the object $o$ %
        does not evolve to the object $\no$ in the $\env$ membrane, 
          in any computation.
    Via Lemma~\ref{lem:many_steps}, this language is well-defined, 
      i.e.\ can defined in terms of $o$ and $|x|$.
    Also, Turing machine $\tallyMachine$ from Lemma~\ref{lem:unimem_subset_dttTallyNL} 
      (that solves reachability)
      can be modified~\cite{Imm1988p,Sze1988p} to give $\tallyMachine'$ (that solves unreachability)
      that accepts the language $T'$.
    That is,
        $\tallyMachine'$ accepts if no object with the desired,
          and easy to check, property can be evolved by rule applications. 
    
      In Lemma~\ref{lem:unimem_subset_dttTallyNL} 
        we defined the function $\tau \in \FAC^0$, 
        that maps from $\{0,1\}^*$ to the set of tuples of unary words.
      Recall that $\tau(x)$ maps to a list that contains 
        a unary string 
        $1^{\langle o, |x|\rangle}$ for each $o$ in  $\encoder(x)$. 
    We now prove that $\tau$ is a conjunctive reduction from $L$ to $T'$.   

    Assume $x \in L$,
        this implies that no object in
        $\Pi_{|x|}$ with input $\encoder(x)$ eventually evolves 
        to  $\no$  in the $\env$ membrane. %
 Hence $x \in L$ implies that $\forall w \in \tau(x), w \in T'$.
    
    Assume $x \notin L$,
        this implies that \emph{at least one} object in the initial configuration of $\Pi_{|x|}(\encoder(x))$ eventually evolves  
        a $\no$ object in the $\env$ membrane in each computation of $\Pi_{|x|}$.
    Hence $x \notin L$ implies that $\exists w \in \tau(x)$ such that $w \notin T'$.
\end{proof}

  \begin{lemma}
     $\Reducible{m}{\FAC^0}{\tally\NL}
     \subseteq
     \UniformMembrane$, for acknowledger and \emphatic membranes systems.
     \label{lem:strict_mTallyNL_subset_uniMem}
  \end{lemma}
  \begin{proof}
    Let $L \in \Reducible{m}{\FAC^0}{\tally\NL}$.
    That is,
      there exists a unary language $T \subseteq \{1\}^*$ 
        that is recognised by 
          non-deterministic logspace Turing machine
            $\tallyMachine$,
      and a function $r \in \FAC^0$ that maps
        $x \in \{ 0,1 \}^*$ to 
          a unary word such that $r(x) \in T$ iff $x \in L$.
    Let $q(|x|)  = \max(\{ |r(w)| \mid w \in \{0,1\}^{|x|} \})$,
      that is, the largest word produced by $r$ on any input of length $|x|$. 
    Note that $q(|x|)$ is computable by $\former$ since $r \in \FAC^0$. 

    We present an $\FAC^0$ uniform polynomial-time $\AMowoDiss$ membrane family
      $\membraneFamily$ that recognises $L$.
    The family is composed of two functions:
        $\former\colon \{1\}^* \rightarrow \mathbf{\Pi}$,
      and $\encoder$ that maps each binary word to a multiset of objects 
      from the appropriate member of $\membraneFamily$. 

    Each member $\Pi_{|x|} = \former(1^{|x|})$ of $\membraneFamily$ has one
    single membrane, $\env$, that is both the environment and the input
    membrane.
    On input  $1^{|x|}$
      the function 
    $\former$ produces one configuration graph 
      $G_{\tallyMachine,u}$ 
        for 
        machine $\tallyMachine$
      (that accepts $T$) on each input $1^u$, $1 \leq u \leq q(|x|)$, 
      and one configuration graph $G_{N,u}$ for machine $N_\tallyMachine$
      (that accepts the compliment of $T$) on each input $1^u$, $1 \leq u \leq
      q(|x|)$. 
    (Note that this is a generalization of the technique used in the proof of
        Lemma~\ref{lem:NL-sub-PMCwodiss}.) 
   
    Each of the $2q(|x|)$ configuration graphs are
      modified to give a set of rules and objects of 
      a single membrane system
      $\Pi_{|x|}$ using the same technique as used in the proof of
        Lemma~\ref{lem:NL-sub-PMCwodiss}.
    In summary, the vertices of the configuration graphs become objects in $\Pi_{|x|}$ 
      and the edges in the graph become type~($\text{a}$) rules.
    There is a rule mapping the object encoding the accepting configuration of 
      $\tallyMachine$ to $\yes$ 
      and rule mapping object encoding the accepting configuration of $N_\tallyMachine$ to $\no$. 
    Both $\tallyMachine$ and $N_\tallyMachine$ are logspace machines and 
    so their configuration graphs are of polynomial size and so the membrane 
    system is of polynomial size.
    It is relatively straightforward to verify that $\former \in \FAC^0$.
    
     The input encoder $\encoder(x)$ simulates 
        $r(x)$ to find $1^u$, then outputs two objects $c_{i,u}$ and $c_{j,u}$
        which encode the vertex of the configuration graph corresponding to
        the initial configurations of Turing machines $\tallyMachine$ and
        $N_\tallyMachine$ respectively on input $1^u=r(x)$. 
    Since $r \in \FAC^0$ it is not difficult to see that $\encoder \in \FAC^0$.

    We now show that 
      the membrane system $\Pi_{|x|}$ on input $\encoder(x)$ accepts if $x \in L$ and
      otherwise rejects. 

      Suppose $x \in L$. 
      This implies that the word $r(x) = 1^u$ is in the tally set $T$ 
        and so at least one computation of $\tallyMachine$ accepts $1^u$.
      It also implies that there is no computation of $N_\tallyMachine$
        that accepts on input $1^u$.
      The input membrane of $\Pi_{|x|}$ contains $\encoder(x)$ 
        which includes $c_{i,u}$ encoding the configuration graph vertex that represents 
        the initial configuration of Turing machine $\tallyMachine$ on input
        $1^u$.
       In the proof of Lemma~\ref{lem:NL-sub-PMCwodiss} we show 
         how $\Pi_{|x|}$  has the property that  
         there is a sequence of rules from the input object $c_{i,u}$ to the 
         $\yes$ object and so 
          $\Pi_{|x|}(\encoder(x))$ will accept.
       Likewise there is no path from $c_{j,u}$ to $\no$. 

      Suppose $u \notin L$.
      This implies that the word $r(u) = 1^j$ is not in the tally set $T$ 
        and that therefore there is no accepting configuration of
        $\tallyMachine$
        on input $1^u$, however, there is at least one accepting computation
        of~$N_\tallyMachine$ on the same input. 
       In the proof of Lemma~\ref{lem:NL-sub-PMCwodiss} we show 
         how the construction of $\Pi_n$ is such that 
         there is a sequence of rules from the input object $c_{j,u}$ to the 
         $\no$ object and so 
          $\Pi_n(\encoder(x))$ will reject. 
       Likewise there is no path from $c_{i,u}$ to $\yes$. 

      Therefore the pair of functions $\former$
        and $\encoder$ provide a uniform family of polynomial time 
        $\AMowoDiss$ membrane systems that accept any language in
        $\Reducible{m}{\FAC^0}{\tally\NL}$.
  \end{proof}

  \section{Open Problems}
    \label{sec:openquestions}
    \paragraph{The power of recogniser membrane systems.}
    In Sections~\ref{sec:semi_uniform} and~\ref{sec:membrane-uni} of this
    paper we characterise the power of acknowledger membrane systems
    (Definition~\ref{def:acknowledger}), which are a generalisation of
    recogniser membrane systems.  In Section~\ref{sec:stricter} we give upper
    and lower bounds on the power of the more restricted \emphatic membrane
    systems (Definition~\ref{def:emphaticrecogniser}), which are closer in
    power to standard recogniser membrane systems. We also give upper bounds
    on the power of uniform and semi-uniform recogniser membrane systems
    (Definition~\ref{def:recogniser}), as well as showing that these classes
    are distinct. 
    
    However, we have not characterised the power of $\AMowoDiss$
    recogniser membrane systems (Definition~\ref{def:recogniser}) with the kind of tight uniformity conditions used in this paper.
    In such systems, in an accepting computation 
      exactly one $\yes$ object,
      or in a rejecting computation exactly one $\no$ object,
      is produced at the final step.
    A consequence of this is that our techniques for showing lower bounds on
    the power of acknowledger and \emphatic systems
    (Sections~\ref{sec:semi_uniform},~\ref{sec:membrane-uni}
    and~\ref{sec:stricter}) in terms of non-deterministic logspace-bounded
    Turing machines do not immediately carry over to recogniser systems. 
       
    As future work, we suggest that recogniser systems could be characterised
    via \emph{unambiguous} non-deterministic logspace-bounded Turing
    machines~\cite{AL1998}. An unambiguous machine accepts an input if and
    only if it has exactly one accepting computation. 
    Perhaps the class of problems solved by 
      semi-uniform families of recogniser $\AMowoDiss$ systems,
      i.e.\ $\SemiUniMembrane$, does not contain $\NL$-complete problems
      since the system cannot
        control how many $\yes$ objects it produces?
    Perhaps these semi-uniform recogniser systems can solve 
        $s$-$t$ connectivity for ``mangrove'' graphs,
        i.e.\ graphs where there is exactly one path between 
          each pair of vertices which is contained in
          unambiguous logspace~\cite{AL1998}?
           Formally, we conjecture that $\SemiUniMembrane = \ComplexityFont{RUSPACE}(\log n)$~\cite{AL1998}. 
    We also conjecture that for the analogous uniform families of recogniser systems 
      $\UniformMembrane = \Reducible{m}{\FAC^0}{\ComplexityFont{RUSPACE}(\log
        n)}$. 
         If proven, our conjectures, taken together with previous results~\cite{AL1998}, would give a restatement of the relationship between the classes $\L$, unambiguous logspace and $\NL$ in the membrane computing
    model, as well as the other classes shown in Figure~\ref{fig:summaryofallresults}. 
    
    \paragraph{Tight uniformity conditions for other classes of membrane
      systems.}
    In this paper and others~\cite{muphthesis, MW2007p, MW2008c, MW2008bc, MW2011}, 
    we have put forward the idea of exploring the power of membrane systems
    under tight uniformity conditions. Others have since carried on this line
    of investigation~\cite{PLMZ2013}. Besides the main result in this paper 
    (exhibiting systems where uniformity is a strictly weaker notion than
    semi-uniformity) this has led to various other characterisations of the
    power of a variety of classes of membrane systems and a teasing apart of
    their power.  A number of other varieties of membrane systems
    (e.g.~\cite{GPR2009,PP2010}) characterise the complexity class $\P$, but
    where the lower-bound actually depends on the use of $\P$ uniformity. As
    future work, it would be interesting to investigate these, and other,
    systems under suitably tight notions of uniformity or semi-uniformity. 

    \paragraph{Upper-bounding $\tally\NL$.}
    While we know that $\tally\NL \subsetneq \NL$ it would be interesting to find other classes to upper bound  $\tally\NL$. 
    It is known that if a sparse language is complete for $\NL$ then $\NL
    \subseteq  
    \DLOGTIME$ uniform-$\TC^0$~\cite{CS2000,HAB2002}.
    Is it possible to show that $\tally\NL \subseteq \DLOGTIME$ uniform-$\TC^0$?

  \paragraph{Classes reducible to $\tally\NL$.}
  We conjecture that $\Reducible{m}{\FAC^0}{\tally\NL}$ is strictly 
  contained in $\Reducible{dtt}{\FAC^0}{\tally\NL}$.
  Giving an exact characterisation of \emphatic membrane systems studied in Section~\ref{sec:stricter}
  may provide some insights into this. 
  We also conjecture that 
    $\Reducible{dtt}{\FAC^0}{\tally\NL} \neq
    \Reducible{ctt}{\FAC^0}{\tally\NL}$.
  A lead to solve this may come from Ko~\cite{Ko1989} who showed that  
    $\Reducible{ctt}{\P}{\tally} \neq \Reducible{dtt}{\P}{\tally}$.

  \section*{Acknowledgements}
  Thanks to an anonymous reviewer for a thorough
  reading and helpful comments. This paper appears in a special issue
  dedicated to Mario de Jes\'{u}s P\'{e}rez-Jim\'{e}nez on this 65th birthday; we thank
  Mario for giving us the opportunity to enjoy Sevilla and its wonderful
  cortados.


\begin{thebibliography}{10}

\bibitem{Adl1994p}
Adleman, L.: Molecular computation of solutions to combinatorial problems,
\newblock \emph{Science}, \textbf{266}, 1994, 1021--1024.

\bibitem{AMP2003p}
Alhazov, A., Mart{\'i}n-Vide, C., Pan, L.: Solving a {PSPACE}-Complete Problem
  by Recognizing {P Systems} with Restricted Active Membranes,
\newblock \emph{Fundamenta Informaticae}, \textbf{58}(2), 2003, 67--77.

\bibitem{AP2004p}
Alhazov, A., Pan, L.: Polarizationless {P Systems} with active membranes,
\newblock \emph{Grammars}, \textbf{7}, 2004, 141--159.

\bibitem{AP2007c}
Alhazov, A., P\'{e}rez-Jiménez, M.~J.: Uniform Solution to {QSAT} Using
  Polarizationless Active Membranes,
\newblock in: \emph{Machines, Computations and Universality (MCU)}, vol. 4664
  of \emph{Lecture Notes in Computer Science}, Springer, 2007,  122--133.

\bibitem{AK2010}
Allender, E., Kouck\'{y}, M.: Amplifying lower bounds by means of
  self-reducibility,
\newblock \emph{Journal of the ACM}, \textbf{57}, March 2010, 14:1--14:36.

\bibitem{AL1998}
Allender, E., Lange, K.-J.: {RUSPACE}$(\log n) \subseteq $\ {DSPACE}$(\log^{2}
  n / \log \log n)$,
\newblock \emph{Theory of Computing Systems}, \textbf{31}(5), 1998, 539--550.

\bibitem{AB2009x}
Arora, S., Barak, B.: \emph{Computational Complexity: A Modern Approach},
\newblock Cambridge University Press, 2009,
\newblock ISBN 978-0-511-53381-5.

\bibitem{Winfree05bar}
Barish, R.~D., Rothemund, P. W.~K., Winfree, E.: Two computational primitives
  for algorithmic self-assembly: copying and counting,
\newblock \emph{Nano Lett.}, \textbf{5}, 2005, 2586--2592.

\bibitem{BSRW2009}
Barish, R.~D., Schulman, R., Rothemund, P. W.~K., Winfree, E.: An
  information-bearing seed for nucleating algorithmic self-assembly,
\newblock \emph{{PNAS}}, \textbf{106}(15), 2009, 6054--6059.

\bibitem{BB1988}
Book, R.~V., Ko, K.-I.: On sets truth-table reducible to sparse sets,
\newblock \emph{SIAM Journal of Computing}, \textbf{17}(5), 1988, 903--919.

\bibitem{Bor1977}
Borodin, A.: On relating time and space to size and depth,
\newblock \emph{SIAM Journal on Computing}, \textbf{6}(4), 1977, 733--744.

\bibitem{BHL1995}
Buhrman, H., Hemaspaandra, E., Longpre, L.: {SPARSE} reduces conjunctively to
  {TALLY},
\newblock \emph{SIAM Journal of Computing}, \textbf{24}, June 1995, 673--681.

\bibitem{CS2000}
Cai, J.-Y., Sivakumar, D.: Resolution of {H}artmanis' conjecture for {NL}-hard
  sparse sets,
\newblock \emph{Theoretical Computer Science}, \textbf{240}(2), 2000, 257 --
  269.

\bibitem{CSV1984}
Chandra, A.~K., Stockmeyer, L.~J., Vishkin, U.: Constant depth reducibility,
\newblock \emph{SIAM Journal of Computing}, \textbf{13}(2), 1984, 423--439.

\bibitem{nubotsDNA20}
Chen, H.-L., Doty, D., Holden, D., Thachuk, C., Woods, D., Yang, C.-T.: Fast
  algorithmic self-assembly of simple shapes using random agitation,
\newblock \emph{DNA20: The 20th International Conference on DNA Computing and
  Molecular Programming}, LNCS, Springer, Kyoto, Japan, September 2014.

\bibitem{CHMT2012}
Condon, A., Hu, A.~J., Ma{\v{n}}uch, J., Thachuk, C.: Less haste, less waste:
  on recycling and its limits in strand displacement systems,
\newblock \emph{Journal of the Royal Society -- Interface focus},
  \textbf{2}(4), 2012, 512--521.

\bibitem{Cook2009}
Cook, M., Soloveichik, D., Winfree, E., Bruck, J.: Programmability of 
  {C}hemical {R}eaction {N}etworks,
\newblock in: \emph{Algorithmic Bioprocesses}, Springer, 2009,  543--584.

\bibitem{dabbyChenSODA2012}
Dabby, N., Chen, H.-L.: Active self-assembly of simple units using an insertion
  primitive,
\newblock \emph{SODA: Proceedings of the Twenty-fourth Annual ACM-SIAM
  Symposium on Discrete Algorithms}, 2012.

\bibitem{FSS1984p}
Furst, M.~L., Saxe, J.~B., Sipser, M.: Parity, circuits and the polynomial-time
  hierarchy,
\newblock \emph{Theory of Computing Systems (formerly Mathematical Systems
  Theory)}, \textbf{17}(1), 1984, 13--27.

\bibitem{GPR2009}
Guti{\' e}rrez-Escudero, R., P{\' e}rez-Jim{\' e}nez, M.~J., Rius-Font, M.:
  Characterizing tractability by tissue-like {P} systems,
\newblock in: \emph{Workshop on Membrane Computing 10}, vol. 5957 of
  \emph{Lecture Notes in Computer Science}, Springer, 2009,  269--181.

\bibitem{NJNC2006p}
Guti\'{e}rrez-Naranjo, M.~A., P\'{e}rez-Jim\'{e}nez, M.~J.,
  Riscos-N\'{u}{\~n}ez, A., Romero-Campero, F.~J.: Computational efficiency of
  dissolution rules in membrane systems,
\newblock \emph{International Journal of Computer Mathematics}, \textbf{83}(7),
  2006, 593--611.

\bibitem{HRBBLS2000}
Head, T., Rozenberg, G., Bladergroen, R.~S., Breek, C. K.~D., Lommerse, P.
  H.~M., Spaink, H.~P.: Computing with {DNA} by operating on plasmids,
\newblock \emph{Biosystems}, \textbf{57}(2), 2000, 87--93.

\bibitem{HAB2002}
Hesse, W., Allender, E., {Mix Barrington}, D.~A.: Uniform constant-depth
  threshold circuits for division and iterated multiplication,
\newblock \emph{Journal of Computer and System Sciences}, \textbf{65}(4), 2002,
  695--716.

\bibitem{Imm1988p}
Immerman, N.: Nondeterministic space is closed under complementation,
\newblock \emph{SIAM Journal of Computing}, \textbf{17}(5), 1988, 935--938.

\bibitem{Imm1999x}
Immerman, N.: \emph{Descriptive Complexity},
\newblock Springer, 1999,
\newblock ISBN 0387986006.

\bibitem{Ko1989}
Ko, K.-I.: Distinguishing conjunctive and disjunctive reducibilities by sparse
  sets,
\newblock \emph{Information and Computation}, \textbf{81}(1), 1989, 62--87.

\bibitem{LLS1975}
Ladner, R.~E., Lynch, N.~A., Selman, A.~L.: A comparison of polynomial time
  reducibilities,
\newblock \emph{Theoretical Computer Science}, \textbf{1}(2), 1975, 103--123.

\bibitem{Lip1995}
Lipton, R.~J.: {DNA solution of hard computational problems},
\newblock \emph{Science}, \textbf{268}(5210), 1995, 542--545.

\bibitem{LWFCCS2000}
Liu, Q., Wang, L., Frutos, A.~G., Condon, A.~E., Corn, R.~M., Smith, L.~M.:
  {DNA} computing on surfaces,
\newblock \emph{Nature}, \textbf{403}(6766), 2000, 175--179.

\bibitem{MalchikWinsow}
Malchik, C., Winslow, A.: Tight Bounds for Active Self-Assembly Using an
  Insertion Primitive,
\newblock \emph{ESA: The 22nd European Symposium on Algorithms}, 2014,
\newblock Accepted. Arxiv preprint
  \href{http://arxiv.org/abs/1401.0359}{\texttt{arXiv:1401.0359}} [cs.FL].

\bibitem{GLPZ2013}
Mauri, G., Leporati, A., Porreca, A.~E., Zandron, C.: Recent
  complexity-theoretic results on {P} systems with active membranes,
\newblock \emph{Journal of Logic and Computation}, 2013, (awaiting
  publication).

\bibitem{BIS1990p}
{Mix Barrington}, D.~A., Immerman, N., Straubing, H.: On uniformity within
  {NC$^1$},
\newblock \emph{Journal of Computer and System Sciences}, \textbf{41}(3), 1990,
  274--306.

\bibitem{muphthesis}
Murphy, N.: \emph{Uniformity conditions for membrane systems: {U}ncovering
  complexity below {P}},
\newblock Ph.D. Thesis, National University of Ireland Maynooth, 2010.

\bibitem{MW2007p}
Murphy, N., Woods, D.: Active membrane systems without charges and using only
  symmetric elementary division characterise {P},
\newblock in: \emph{8th International Workshop on Membrane Computing}, vol.
  4860 of \emph{Lecture Notes in Computer Science}, Springer, 2007,  367--384.

\bibitem{MW2008c}
Murphy, N., Woods, D.: A characterisation of {NL} using membrane systems
  without charges and dissolution,
\newblock in: \emph{Unconventional Computing, 7th International Conference, UC
  2008, Vienna, Austria,}, vol. 5204 of \emph{Lecture Notes in Computer
  Science}, Springer, 2008,  164--176.

\bibitem{MW2008bc}
Murphy, N., Woods, D.: On acceptance conditions for membrane systems:
  characterisations of {L} and {NL},
\newblock \emph{Proceedings of the International Workshop on The Complexity of
  Simple Programs}, 1, Electronic Proceedings in Theoretical Computer Science,
  2009,
\newblock Arxiv preprint:
  \href{http://arxiv.org/abs/arXiv:0906.3327v1}{\texttt{arXiv:0906.3327v1}}
  [cs.CC].

\bibitem{MW2011}
Murphy, N., Woods, D.: The computational power of membrane systems under tight
  uniformity conditions,
\newblock \emph{Natural Computing}, \textbf{10}(1), 2011, 613--632.

\bibitem{MurphyWoods2013}
Murphy, N., Woods, D.: {AND} and/or {OR}: Uniform polynomial-size circuits,
\newblock \emph{MCU: Proceedings of Machines, Computations and Universality},
  128, Electronic Proceedings in Theoretical Computer Science, 2013,
\newblock Arxiv preprint:
  \href{http://arxiv.org/abs/1212.3282v2}{\texttt{arXiv:1212.3282v2}} [cs.CC].

\bibitem{OKLL1997}
Ouyang, Q., Kaplan, P.~D., Liu, S., Libchaber, A.: {DNA} solution of the
  Maximal Clique Problem,
\newblock \emph{Science}, \textbf{278}(5337), 1997, 446--449.

\bibitem{PP2010}
Pan, L., P\'{e}rez-Jim\'{e}nez, M.~J.: Computational complexity of tissue-like
  P systems,
\newblock \emph{Journal of Complexity}, \textbf{26}(3), 2010, 296--315.

\bibitem{Pap1993x}
Papadimitriou, C.~H.: \emph{Computational Complexity},
\newblock Addison Wesley, 1993,
\newblock ISBN 0201530821.

\bibitem{Par1994b}
Parberry, I.: \emph{Circuit complexity and neural networks},
\newblock MIT Press, 1994,
\newblock ISBN 0-262-16148-6.

\bibitem{PRS2010_handbook}
P{\u a}un, G., Rozenberg, G., Salomaa, A., Eds.: \emph{The Oxford Handbook of
  Membrane Computing},
\newblock Oxford University Press, Inc., New York, NY, USA, 2010,
\newblock ISBN 0199556679, 9780199556670.

\bibitem{PRRW2009x}
P{\'e}rez-Jim{\'e}nez, M.~J., Riscos-N{\' u}{\~ n}ez, A., Romero-Jim{\'e}nez,
  A., Woods, D.: \emph{The Oxford Handbook of Membrane systems}, chapter 12:
  Complexity -- Membrane Division, Membrane Creation,
\newblock Oxford University Press, 2009,  302--336.

\bibitem{PRS2003p}
P\'{e}rez-Jim\'{e}nez, M.~J., Romero-Jim\'{e}nez, A., Sancho-Caparrini, F.:
  Complexity classes in models of cellular computing with membranes,
\newblock \emph{Natural Computing}, \textbf{2}(3), 2003, 265--285.

\bibitem{PLMZ2013}
Porreca, A.~E., Leporati, A., Mauri, G., Zandron, C.: Sublinear-space {P}
  systems with active membranes,
\newblock in: \emph{Proceedings of the 13th International Conference on
  Membrane Computing 2012}, vol. 7762 of \emph{Lecture Notes in Computer
  Science}, Springer, 2013,  342--357.

\bibitem{post1944}
Post, E.~L.: Recursively enumerable sets of positive integers and their
  decision problems,
\newblock \emph{Bulletin of the American Mathematical Society}, \textbf{50}(5),
  1944, 284--316.

\bibitem{Pau2003}
P\u{a}un, G.: Membrane computing,
\newblock in: \emph{Fundamentals of computation theory}, vol. 2751 of
  \emph{Lecture Notes in Computer Science}, Springer, 2003,  284--295.

\bibitem{Pau2005c}
P\u{a}un, G.: Further twenty six open problems in membrane computing,
\newblock \emph{Proceedings of the {T}hird {B}rainstorming {W}eek on {M}embrane
  {C}omputing}, F\'{e}nix Editoria, 2005.

\bibitem{qian2011scaling}
Qian, L., Winfree, E.: Scaling up digital circuit computation with {DNA} strand
  displacement cascades,
\newblock \emph{Science}, \textbf{332}(6034), 2011, 1196.

\bibitem{rothemund2000program}
Rothemund, P.~W., Winfree, E.: The program-size complexity of self-assembled
  squares,
\newblock \emph{Proceedings of the thirty-second annual ACM symposium on Theory
  of computing}, ACM, 2000.

\bibitem{soloveichik2008computation}
Soloveichik, D., Cook, M., Winfree, E., Bruck, J.: Computation with finite
  stochastic chemical reaction networks,
\newblock \emph{Natural Computing}, \textbf{7}(4), 2008, 615--633.

\bibitem{WS2005p}
Soloveichik, D., Winfree, E.: The computational power of {B}enenson automata,
\newblock \emph{Theoretical Computer Science}, \textbf{344}, 2005, 279--297.

\bibitem{SW2007p}
Soloveichik, D., Winfree, E.: Complexity of self-assembled shapes,
\newblock \emph{SIAM Journal of Computing}, \textbf{36}(6), 2007, 1544--1569.

\bibitem{Sos2003p}
Sos\'{i}k, P.: The computational power of cell division in {P systems}: Beating
  down parallel computers?,
\newblock \emph{Natural Computing}, \textbf{2}(3), 2003, 287--298.

\bibitem{Sze1988p}
Szelepcs{\'e}nyi, R.: The method of forced enumeration for nondeterministic
  automata,
\newblock \emph{Acta Informatica}, \textbf{26}(3), 1988, 279--284.

\bibitem{TC2012}
Thachuk, C., Condon, A.: Space and energy efficient computation with DNA strand
  displacement systems,
\newblock in: \emph{18th International Conference on DNA Computing and
  Molecular Programming}, vol. 7433 of \emph{Lecture Notes in Computer
  Science}, Springer, 2012,  135--150.

\bibitem{thachuk2013space}
Thachuk, C.~J.: \emph{Space and energy efficient molecular programming and
  space efficient text indexing methods for sequence alignment},
\newblock Ph.D. Thesis, University of British Columbia, 2013.

\bibitem{Vol1999x}
Vollmer, H.: \emph{Introduction to Circuit Complexity: A Uniform Approach},
\newblock Springer-Verlag New York, Inc., Secaucus, NJ, USA, 1999,
\newblock ISBN 3540643109.

\bibitem{nubots}
Woods, D., Chen, H.-L., Goodfriend, S., Dabby, N., Winfree, E., Yin, P.: Active
  self-assembly of algorithmic shapes and patterns in polylogarithmic time,
\newblock \emph{ITCS'13: Proceedings of the 4th conference on Innovations in
  Theoretical Computer Science}, ACM, 2013,
\newblock Full version:
  \href{http://arxiv.org/abs/1301.2626}{\texttt{arXiv:1301.2626}} [cs.DS].

\end{thebibliography}
\end{document}